\documentclass[english]{article}
\usepackage{amsthm}
\usepackage{amsmath}
\usepackage{amssymb}
\usepackage[unicode=true,pdfusetitle,
 bookmarks=true,bookmarksnumbered=false,bookmarksopen=false,
 breaklinks=false,pdfborder={0 0 1},backref=false,colorlinks=false]
 {hyperref}
\usepackage{breakurl}

\makeatletter

\providecommand{\tabularnewline}{\\}

\theoremstyle{plain}
\newtheorem{thm}{\protect\theoremname}
  \theoremstyle{definition}
  \newtheorem{defn}[thm]{\protect\definitionname}
  \theoremstyle{plain}
  \newtheorem{cor}[thm]{\protect\corollaryname}
  \theoremstyle{plain}
  \newtheorem{lem}[thm]{\protect\lemmaname}
  \theoremstyle{plain}
  \newtheorem{prop}[thm]{\protect\propositionname}
  \theoremstyle{remark}
  \newtheorem{rem}[thm]{\protect\remarkname}

\usepackage{bbm}
\usepackage{graphicx}

\makeatother

  \providecommand{\corollaryname}{Corollary}
  \providecommand{\definitionname}{Definition}
  \providecommand{\lemmaname}{Lemma}
  \providecommand{\propositionname}{Proposition}
  \providecommand{\remarkname}{Remark}
\providecommand{\theoremname}{Theorem}

\begin{document}
\global\long\def\one{\mathbbm{1}}

\global\long\def\id{\mathrm{id}}
\global\long\def\diag{\mathrm{diag}}
\global\long\def\tr{\mathrm{Tr\,}}
\global\long\def\e{\mathbb{E}}

\global\long\def\rhoin{\rho_{\mathrm{in}}}
\global\long\def\rhoout{\rho_{\mathrm{out}}}

\global\long\def\C{\mathbb{C}}
\global\long\def\R{\mathbb{R}}
\global\long\def\Z{\mathbb{Z}}
\global\long\def\N{\mathbb{N}}

\global\long\def\Ncal{\mathcal{N}}
\global\long\def\A{\mathcal{A}}
\global\long\def\D{\mathcal{D}}
\global\long\def\Q{\mathcal{Q}}
\global\long\def\Hcal{\mathcal{H}}
\global\long\def\Ccal{\mathcal{C}}
\global\long\def\Ecal{\mathcal{E}}
\global\long\def\S{\mathcal{S}}
\global\long\def\M{\mathcal{M}}

\global\long\def\Qmax{\mathcal{Q}_{\mathrm{max}}}

\global\long\def\Tr#1{\mathrm{Tr}\left(#1\right)}
\global\long\def\Wg#1{\mathrm{Wg}\left(#1\right)}
\global\long\def\E#1{\mathrm{\mathbb{E}}\left[#1\right]}

\global\long\def\bra#1{\langle#1|}
\global\long\def\ket#1{|#1\rangle}
\global\long\def\braket#1#2{\left\langle #1\middle|#2\right\rangle }

\global\long\def\Mtwo#1#2#3#4{\begin{pmatrix}#1  &  #2\\
#3  &  #4 
\end{pmatrix}}

\global\long\def\Deltaln{\Delta_{\lambda}^{\otimes n}}
\global\long\def\An{\mathcal{A}^{\otimes n}}
\global\long\def\Srreg{\overline{S}_{r}^{\mathrm{reg}}}

\global\long\def\fracp#1#2{\left(\frac{#1}{#2}\right)}

\global\long\def\ot{\otimes}

\global\long\def\reg{\mathrm{reg}}

\title{Average output entropy for quantum channels}

\author{Christopher King\textsuperscript{1} and David K. Moser\textsuperscript{1,2}{\small }\\
{\small }\\
{\small 1: Department of Mathematics}\\
{\small 2: Department of Physics}\\
{\small Northeastern University}\\
{\small Boston MA 02115}}
\maketitle
\begin{abstract}
We study the regularized average Renyi output entropy $\overline{S}_{r}^{\reg}$
of quantum channels. This quantity gives information about the average
noisiness of the channel output arising from a typical, highly entangled
input state in the limit of infinite dimensions. We find a closed
expression for $\beta_{r}^{\reg}$, a quantity which we conjecture
to be equal to $\Srreg$. We find an explicit form for $\beta_{r}^{\reg}$
for some entanglement-breaking channels, and also for the qubit depolarizing
channel $\Delta_{\lambda}$ as a function of the parameter $\lambda$.
We prove equality of the two quantities in some cases, in particular
we conclude that for $\Delta_{\lambda}$ both are non-analytic functions
of the variable $\lambda$.
\end{abstract}
\tableofcontents{}

\section{Introduction}

The noisiness of a quantum channel is closely related to its ability
to transfer information, and is reflected in the values of the various
channel capacities. Much work has been done on understanding the capacities,
for example \cite{smith_quantum_channel_capacity} provides a recent
survey. These capacities are sometimes difficult to analyze directly,
for example the Holevo capacity is computed using multiple output
states. Accordingly other more mathematically tractable quantities
have been used to measure the amount of noise introduced by the channel.
One example is the minimal output Renyi entropy \cite{amosov_holevo_werner_additivity}
of a channel ${\cal A}$, defined for $r\ge1$: 
\[
S_{r,\min}(\A)=\min_{\ket{\phi}}\,\frac{1}{1-r}\,\log\tr\left(\A(\ket{\phi}\bra{\phi})\right)^{r}
\]
At $r=1$ this yields the minimal output von Neumann entropy, which
has a close connection to the classical capacity of the channel \cite{matsumoto_shimono_winter_additivity_remarks}.
The entropies for $r>1$ also provide useful properties of the channel,
and in some cases are easier to analyze and compute. The famous additivity
conjecture concerns the regularized version of this quantity, which
is defined as 
\[
S_{r,\min}^{\reg}(\A)=\lim_{n\rightarrow\infty}\,\frac{1}{n}\, S_{r,\min}(\A^{\ot n})
\]
(the existence of the limit is an easy consequence of the sub-additivity
bound ${\cal S}_{r,\min}({\cal A}\ot{\cal B})\leq{\cal S}_{r,\min}({\cal A})+{\cal S}_{r,\min}({\cal B})$).
While the inequality $S_{r,\min}^{\reg}(\A)\le S_{r,\min}(\A)$ is
always true, it was an open question for several years whether equality
holds. It is now known that equality does not hold in general \cite{hayden_winter_p_counterexamples,hastings_counterexample},
so this raises the interesting question of determining $S_{r,\min}^{\reg}(\A)$.
Except for those channels where additivity does hold, the value of
this regularized quantity is unknown. For channels with non-additive
Holevo capacity the classical capacity is also defined by such a regularized
quantity, so it is an important problem to find new ways to calculate
these regularized limits. In a sense we follow a strategy opposite
to the random channel methods used to disprove the additivity conjecture;
our high-dimensional channels are products of fixed channels, and
thus are constructed explicitly.

\medskip{}
In the hopes of finding some new insights into these regularized channel
properties we consider a related quantity which also measures the
noisiness of the channel, namely the average output Renyi entropy.
This measures the entropy of the channel output for typical input
states, rather than the smallest value which is used to compute the
minimal output entropy. For a finite-dimensional channel this is defined
for all $r\ge1$ by 
\[
\overline{S}_{r}(\A)=\E{\frac{1}{1-r}\,\log\tr\left(\A(\ket{\phi}\bra{\phi})\right)^{r}}
\]
where the expectation is computed using the uniform probability measure
on the set of input pure states. We also consider the related quantity
\[
\beta_{r}(\A)=\,\frac{1}{1-r}\,\log\E{\tr\left(\A(\ket{\phi}\bra{\phi})\right)^{r}}
\]
Note that by Jensen's inequality 
\begin{equation}
\overline{S}_{r}(\A)\ge\beta_{r}(\A)\,.\label{eq:entropy-ineq}
\end{equation}
These quantities can be computed (at least numerically) for any given
channel. In operational terms, they describe the long-run average
output Renyi entropy of the channel for a sequence of random pure
input states. Loosely speaking, they measure the average noisiness
of an output state from the channel.

\medskip{}
As with the minimal Renyi entropy, we also consider the regularized
versions of these quantities. However unlike the minimal Renyi entropy,
the existence of these regularized limits is not obvious, so we define
them conservatively using the $\liminf$: 
\begin{align}
\Srreg(\A) & =\liminf_{n\rightarrow\infty}\,\frac{1}{n}\,\overline{S}_{r}(\A^{\ot n})\label{eq:def-regs}\\
\beta_{r}^{\reg}(\A) & =\liminf_{n\rightarrow\infty}\,\frac{1}{n}\,\beta_{r}(\A^{\ot n})\nonumber 
\end{align}
\medskip{}
We conjecture that the two quantities in \eqref{eq:def-regs} are
equal, however we do not yet have a proof of this for a general channel.
For the specific channels we look at in more detail the quantity $\beta_{r}^{{\rm reg}}({\cal A})$
is given by the expression above with $\liminf$ replaced by $\lim$.
But from the closed expression for $\beta_{r}^{\reg}(\A)$ presented
below there is a possibility for more complex limiting behavior.

For one special class of channels we can compute a simple formula
for $\beta_{r}^{\reg}(\A)$ for integer values of $r$. These are
a subset of the entanglement breaking (E-B) channels \cite{ruskai_entanglementbreaking,ruskai_shor_general_ebc},
which can be written in the form ${\cal A}(\rho)=\sum_{k}\sigma_{k}\tr(X_{k}\rho)$
-- for some states ${\sigma_{k}}$ and with $X_{k}$ a POVM -- meeting
the additional condition $\Tr{\prod_{i=1}^{r}\sigma_{k_{i}}}\geq0$.
For $r=2$ this includes all the E-B channels. For higher $r$ a particular
class of E-B channels that fulfill the condition are the QC channels
as defined by Holevo \cite{holevo_qc_channels}, where $\sigma_{k}=\ket k\bra k$
are pure states formed by an orthonormal basis. For unital E-B channels
which satisfy the condition on the $\sigma_{k}$ we can prove even
more, namely that $\Srreg(\A)$=$\beta_{r}^{\reg}(\A)$ = $\log d$.\medskip{}

We consider only finite-dimensional channels with equal input and
output dimensions, and we define the dimension of the channel to be
this common value. The identity matrix is denoted by $\one$.
\begin{thm}
[proof in \ref{sub:Entanglement-breaking-channels}]\label{thm:beta-reg-eb}

\ \begin{itemize}
\item[(a)]Let $r\geq2$ be an integer, and $\A(\rho)=\sum\sigma_{k}\Tr{X_{k}\rho}$
an entanglement breaking channel satisfying the condition $\Tr{\prod_{i=1}^{r}\sigma_{k_{i}}}\geq0$
for all choices of $\{k_{i}\}$. Then

\[
\beta_{r}^{\reg}(\A)=\lim_{n\to\infty}\frac{1}{n}\beta_{r}(\A^{\ot n})=\frac{1}{1-r}\log\tr({\cal A}(\one/d)^{r})\,.
\]

\item[(b)]Let ${\cal A}$ be a d-dimensional unital entanglement-breaking
channel, $\A(\one)=\one$, satisfying the condition $\Tr{\prod_{i=1}^{m}\sigma_{k_{i}}}$
for all integer $m\geq2$. Then for all real $r\geq1$
\[
\Srreg(\A)=\beta_{r}^{\reg}(\A)=\log d\,.
\]

\end{itemize}
\end{thm}
We also derive an explicit expression for $\beta_{r}^{\reg}(\A)$
in the general case. The statement of this result requires some additional
notation. First recall the definition of the Choi-Jamiolkowski representation
\cite{choi_complete_positivity} of a channel, namely

\[
Choi({\cal A})=\sum_{x,y}{\cal A}(\ket x\bra y)\ot\ket x\bra y
\]
where $\ket x$ and $\ket y$ are orthonormal bases of pure input
states. Also let $Sym(r)$ denote the symmetric group on $r$ letters.
Then every element $\alpha\in Sym(r)$ defines a permutation operator
on $(\C^{d})^{\ot r}$ by ${\cal R}(\alpha)(v_{1}\ot\cdots\ot v_{r})=v_{\alpha(1)}\ot\cdots\ot v_{\alpha(r)}$.
\begin{defn}
\label{def:q}Let ${\cal A}$ be a channel. For all $\alpha\in Sym(r)$
define

\begin{equation}
\Q_{{\cal A},\, r}(\alpha)=\tr\left[Choi({\cal A})^{\ot r}({\cal R}(123\dots r)\ot{\cal R}(\alpha))\right]\,.
\end{equation}

\end{defn}
Furthermore, for some channels $\beta_{r}^{\reg}$ is a simple limit.
One such class are the entrywise positive maps defined studied in
\cite{king_multiplicativity_2005}.
\begin{defn}
A channel $\A$ is called \emph{entrywise positive }if there exist
a bases for input and output space such that $\bra s\A(\ket x\bra y)\ket t\geq0$
for all $x,y,s,t$.
\end{defn}
It is clear that this definition is equivalent to $Choi(\A)$ being
entrywise positive.
\begin{thm}
[proof in \ref{sub:Average-moments}]\label{thm:beta-reg}

\ \begin{itemize}\item[(a)]Let $r\geq2$ be an integer and let $\Q_{\max}=\max_{\alpha\in Sym(r)}|\Q_{{\cal A},\, r}(\alpha)|$.
Then 
\[
\beta_{r}^{\reg}(\A)=\frac{r\log d-\log\Q_{\max}}{r-1}\,.
\]

\item[(b)] Let $r\ge2$ be an integer. If the maximum $\Qmax$ is
attained for a unique $\alpha$ then the $\liminf$ in $\beta_{r}^{\reg}$
can be replaced with a regular limit:

\[
\beta_{r}^{\reg}(\A)=\lim_{n\to\infty}\frac{1}{n}\beta_{r}(\A^{\ot n})\,.
\]

\item[(c)]If the channel $\A$ is entrywise positive then for all
integer $r\geq2$
\[
\beta_{r}^{\reg}(\A)=\lim_{n\to\infty}\frac{1}{n}\beta_{r}(\A^{\ot n})\,.
\]

\end{itemize}
\end{thm}
The evaluation of $\Q_{\max}$ seems to be a difficult problem in
general for large values of $r$. However for two special permutations
the quantity $\Q_{{\cal A},\, r}(\alpha)$ can be evaluated easily,
namely the identity permutation and the full cycle:

\begin{eqnarray*}
\Q_{{\cal A},\, r}(id) & = & \tr{\cal A}(\one)^{r},\quad\quad\Q_{{\cal A},\, r}(123\dots r)=\tr(Choi({\cal A})^{r})\,.
\end{eqnarray*}
Thus for $r=2$ the result can be stated more explicitly as follows.
\begin{cor}
Let ${\cal A}$ be a $d$-dimensional channel, then

\[
\beta_{2}^{\reg}(\A)=2\log d-\log\max[\tr{\cal A}(\one)^{2},\tr(Choi({\cal A})^{2})]
\]

\end{cor}
To make further progress we now focus on one of the simplest cases,
namely the qubit depolarizing channel $\Delta_{\lambda}$ \cite{ruskai_szarek_werner_cptp_analysis},
where we are able to prove a number of additional results. In particular
using concentration of measure arguments \cite{ledoux_measure_concentration,milman_schechtman_asymptotic_theory}
we compute the regularized quantity $\overline{S}_{r}^{\reg}(\Delta_{\lambda})$
-- which we call regularized output entropy in the remainder of the
paper -- for integer values of $r$, and for a range of values of
$\lambda$. One interesting consequence is that this quantity is a
non-analytic function of the depolarizing parameter $\lambda$. Recall
the definition of this channel:

\[
\Delta_{\lambda}(\rho)=\lambda\rho+\frac{1-\lambda}{2}\one
\]
The channel is completely positive for $-1/3\leq\lambda\leq1$ and
is entanglement breaking for $-1/3\leq\lambda\leq1/3$.
\begin{thm}
[proof in \ref{sub:average-output-entropy}]\label{thm:main} 

\ \begin{itemize}\item[(a)]For all $r\in\mathbb{N}$, $r\ge2$,
and $\lambda\in[0,1]$, 
\begin{align*}
\beta_{r}^{\reg}(\Delta_{\lambda}) & =\lim_{n\rightarrow\infty}\,\frac{1}{n}\,\beta_{r}(\Deltaln)=\min\left\{ 1,\frac{2r-\log\left[(1+3\lambda)^{r}+3(1-\lambda)^{r}\right]}{r-1}\right\} 
\end{align*}
in particular
\begin{align*}
\beta_{2}^{\reg}(\Delta_{\lambda}) & =\begin{cases}
1 & \lambda\leq1/\sqrt{3}\\
2-\log(1+3\lambda^{2}) & \lambda>1/\sqrt{3}
\end{cases}\,,\\
\beta_{\infty}^{\reg}(\Delta_{\lambda}) & =\begin{cases}
1 & \lambda\leq1/3\\
2-\log(1+3\lambda) & \lambda>1/3
\end{cases}\,,
\end{align*}
where $\beta_{\infty}^{\reg}(\Delta_{\lambda})=\lim_{r\to\infty}\,\beta_{r}^{\reg}(\Delta_{\lambda})$.

\item[(b)]For all $r\in\mathbb{N}$, $r\ge2$, and $\lambda\in J_{r}$,
\[
\Srreg(\Delta_{\lambda})=\beta_{r}^{\reg}(\Delta_{\lambda})\,,
\]
where $J_{r}=[0,c_{r}]\cup[d_{r},1]\subset[0,1]$ for some $0<c_{r}<d_{r}<1$
(see Table~\ref{tab:validity} in \ref{sub:average-output-entropy}).

\end{itemize}
\end{thm}
\medskip{}

From the explicit form given in (a) it is clear that $\beta_{r}^{\reg}(\Delta_{\lambda})$
does not have a continuous first derivative for some $\lambda_{r}\in[1/3,1/\sqrt{3}]$,
in particular it is non-analytic. Because $\overline{S}_{r}^{\reg}(\Delta_{\lambda})$
is defined as the $\liminf$ of a series upper bounded by 1 it is
well defined for all $0\leq\lambda\leq1$. Furthermore we know from
(b) that it is equal to 1 for some range $\lambda\in[0,c_{r}]$ but
at $\lambda=1$ its value is 0. Therefore $\overline{S}_{r}^{\reg}(\Delta_{\lambda})$
has least one non-analytic point somewhere in the range $\lambda\in[c_{r},d_{r}]$.

It is tempting to associate this non-analyticity with a transition
between distinct phases of the model, but the operational meaning
of this is unclear at the moment. We conjecture that the quantities
$\Srreg(\Delta_{\lambda})$ and $\beta_{r}^{\reg}(\Delta_{\lambda})$
are in fact equal for all $\lambda$, and for all $r\ge1$. It is
noteworthy that the channel $\Delta_{\lambda}$ is entanglement-breaking
at and below the value $\lambda=1/3$.

\medskip{}
Counterexamples to the additivity conjecture have been found so far
by using randomization techniques \cite{hayden_leung_winter_aspectsofentanglement,fukuda_king_random_subspaces,aubrun_szarek_werner_additivity_dvoretzky,hastings_counterexample}.
This has led to an understanding of the behavior of a typical high-dimensional
channel, at least insofar as it affects the minimal output Renyi entropy,
by proving the generic existence of channels all of whose output states
have high entropy. Here we look from a different point of view, by
considering the properties of a typical output state for a product
of many copies of a channel. Open questions remain, for example the
amount of entanglement in a typical output state. We note that the
questions addressed here have a different flavor from arguments based
on locality, since here the system is fully entangled across all copies.

\section{Main result}

\subsection{Notation}

We work with a general channel $\A$ and its tensor product $\Ccal=\A^{\otimes n}$.
The dimension of $\A$ is
\[
d=\dim\A\,
\]
To achieve the results for $\Srreg$ and $\beta_{r}^{\reg}$ in Theorems~\ref{thm:beta-reg-eb},
\ref{thm:beta-reg} and \ref{thm:main} we first find a closed expression
for the average moments for integer $r\ge2$
\begin{equation}
M_{r}(\Ccal)=\E{\Tr{\Ccal(\ket{\phi}\bra{\phi})^{r}}}\,,\label{eq:trace-moment}
\end{equation}
where averaging is over random pure input states $\ket{\phi}=U\ket 0$
with $U$ distributed according to the Haar measure on $SU(d^{n})$.
We then use the relation
\[
\beta_{r}^{\reg}(\A)=\liminf_{n\to\infty}\frac{1}{n(1-r)}\log M_{r}(\A^{\otimes n})
\]

\subsection{Evaluating trace moments}

We rewrite the trace moment \eqref{eq:trace-moment} by inserting
four complete sums that run over the entire input space 
\begin{align}
 & \E{\Tr{\Ccal(\ket{\phi}\bra{\phi})^{r}}}\nonumber \\
 & =\e\Bigg[\tr\bigg(\sum_{a,b,x,y}\ket a\bra a\Ccal\left(\vphantom{\sum}\ket x\braket x{\phi}\braket{\phi}y\bra y\right)\ket b\bra b\bigg)^{r}\Bigg]\nonumber \\
 & =\e\Bigg[\sum_{\substack{\{a_{i},x_{i},y_{i}\}\\
i=1\dots r
}
}\prod_{i=1}^{r}\Ccal_{a_{i}x_{i}y_{i}a_{i+1}}\braket{x_{i}}{\phi}\braket{\phi}{y_{i}}\Bigg]\nonumber \\
 & =\sum_{\substack{\{a_{i},x_{i},y_{i}\}\\
i=1\dots r
}
}\prod_{i=1}^{r}\Ccal_{a_{i}x_{i}y_{i}a_{i+1}}\e\bigg[\prod_{j=1}^{r}\braket{x_{j}}{\phi}\braket{\phi}{y_{j}}\bigg]\label{eq:target-sum}
\end{align}
with the identification $a_{r+1}\equiv a_{1}$, and with 
\[
\Ccal_{axyb}=\bra a\Ccal(\ket x\bra y)\ket b
\]
the matrix elements of the channel. The expectation value in \eqref{eq:target-sum}
\begin{equation}
\e\bigg[\prod_{j=1}^{r}\braket{x_{j}}{\phi}\braket{\phi}{y_{j}}\bigg]=\e\bigg[\prod_{j=1}^{r}\bra{x_{j}}U\ket 0\bra 0U^{*}\ket{y_{j}}\bigg]\label{eq:expectation-value-wein}
\end{equation}
of products of matrix elements of unitaries distributed according
to the Haar measure may be calculated using Weingarten calculus \cite{weingarten}.
The Weingarten function $\mathrm{Wg}:\N\times Sym(r)\to\R$ maps pairs
of dimension $k$ and elements of the symmetric group $Sym(r)$ into
the reals. The general expression is 
\begin{align*}
 & \E{U_{i_{1}j_{1}}\dots U_{i_{r}j_{r}}\overline{U_{i_{1}'j_{1}'}}\dots\overline{U_{i_{r}'j_{r}'}}}\\
 & =\sum_{\alpha,\beta\in Sym(r)}\delta_{i_{\alpha(1)}i_{1}'}\dots\delta_{i_{\alpha(r)}i_{r}'}\delta_{j_{\beta(1)}j_{1}'}\dots\delta_{j_{\beta(r)}j_{r}'}\Wg{k,\beta^{-1}\alpha}\,.
\end{align*}
Thus \eqref{eq:expectation-value-wein} simplifies to
\begin{align}
\e\bigg[\prod_{j=1}^{r}\bra{x_{j}}U\ket 0\bra 0U^{*}\ket{y_{j}}\bigg] & =\sum_{\alpha,\beta\in Sym(r)}\delta_{x_{\alpha(1)}y_{1}}\dots\delta_{x_{\alpha(r)}y_{r}}\Wg{k,\beta\alpha^{-1}}\\
 & =\sum_{\alpha\in Sym(r)}\delta_{x_{\alpha(1)}y_{1}}\dots\delta_{x_{\alpha(r)}y_{r}}C_{k,\, r}\label{eq:evaluated-target-expectation}
\end{align}
where $k$ is the input dimension for ${\cal C}$, and in the last
step the Weingarten function is summed over all permutations $\gamma=\beta\alpha^{-1}$.
This sum can be evaluated explicitly, as was shown for example in
\cite{collins_nechita_graphicalcalculus}: 
\begin{equation}
C_{k,\, r}=\sum_{\gamma\in Sym(r)}\Wg{k,\gamma}=\prod_{j=0}^{r-1}\frac{1}{k+j}\,.\label{eq:comb-factor}
\end{equation}
\medskip{}
We plug the evaluated expectation value \eqref{eq:evaluated-target-expectation}
in our original expression \eqref{eq:target-sum} and get 
\begin{align}
\sum_{\substack{\{a_{i},x_{i},y_{i}\}\\
i=1\dots r
}
} & \bigg(\prod_{i=1}^{r}\Ccal_{a_{i}x_{i}y_{i}a_{i+1}}\cdot\sum_{\alpha\in Sym(r)}\prod_{j=1}^{r}\delta_{x_{\alpha(j)}y_{j}}\cdot C_{k,\, r}\bigg)\nonumber \\
 & =C_{k,\, r}\sum_{\substack{\{a_{i},x_{i}\}\\
i=1\dots r
}
}\sum_{\alpha\in Sym(r)}\prod_{i=1}^{r}\Ccal_{a_{i}x_{i}x_{\alpha(i)}a_{i+1}}\,.\label{eq:trace-moments-simplified}
\end{align}
Now define
\begin{equation}
\Q_{\Ccal,\, r}(\alpha)=\sum_{\substack{\{a_{i},x_{i}\}\\
i=1\dots r
}
}\prod_{i=1}^{r}\Ccal_{a_{i}x_{i}x_{\alpha(i)}a_{i+1}}=\sum_{\substack{\{x_{i}\}\\
i=1\dots r
}
}\tr\prod_{i=1}^{r}\Ccal_{x_{i}x_{\alpha(i)}}\,,\label{eq:q-definition}
\end{equation}
where the matrices $\Ccal_{xy}$ have entries $(\Ccal_{xy})_{ab}=\Ccal_{axyb}$
or more simply $\Ccal_{xy}=\Ccal(\ket x\bra y)$. If it is clear from
context we may omit one or both subscripts $\Q(\alpha)=\Q_{\Ccal}(\alpha)=\Q_{\Ccal,\, r}(\alpha)$.
These are the terms we will analyze to a great length in the rest
of the work.

The $\Q$ defined in this way is identical to the one from Definition~\ref{def:q}.
This can be seen from the following calculation

\begin{align*}
 & \Tr{Choi(\Ccal)^{\ot r}({\cal R}(123\dots r)\ot{\cal R}(\alpha))}\\
= & \sum_{\{a,x\}}\bra{a_{1},x_{1}}\ot\cdots\ot\bra{a_{r},x_{r}}Choi(\Ccal)^{\ot r}\ket{a_{2},x_{\alpha(1)}}\ot\cdots\ot\ket{a_{1},x_{\alpha(r)}}\\
= & \sum_{\{a,x\}}\bra{a_{1}}\Ccal(\ket{x_{1}}\bra{x_{\alpha(1)}})\ket{a_{2}}\bra{a_{2}}\Ccal(\ket{x_{2}}\bra{x_{\alpha(2)}})\ket{a_{3}}\cdots\\
= & \sum_{\substack{\{a_{i},x_{i}\}\\
i=1\dots r
}
}\prod_{i=1}^{r}\Ccal_{a_{i}x_{i}x_{\alpha(i)}a_{i+1}}\\
= & \Q_{\Ccal,\, r}(\alpha)\,.
\end{align*}
In terms of the $\Q_{\Ccal}$ we get our final working expression
for the average moments
\[
M_{r}({\cal C})=C_{k,\, r}\sum_{\alpha\in Sym(r)}\Q_{\Ccal}(\alpha)\,.
\]
In general $\Q_{\Ccal}(\alpha)$ could have complex values. However
in some cases it can be shown to be real. In particular for the depolarizing
channel it follows directly from $\Ccal_{axyb}\geq0$ that $\Q_{\Ccal}(\alpha)$
is positive.

\subsection{Product channels}

In the case where $\Ccal$ is a tensor product $\Ccal=\D\otimes\Ecal$
we work in the product base $\ket x=\ket{x'x''}$. Now the tensor
and channel application are interchangeable
\begin{align*}
\D\otimes\Ecal_{xy} & =\D\otimes\Ecal(\ket{x'x''}\bra{y'y''})\\
 & =\D(\ket{x'}\bra{y'})\otimes\Ecal(\ket{x''}\bra{y''})\\
 & =\D_{x'y'}\otimes\Ecal_{x''y''}\,.
\end{align*}
And therefore, the $\Q_{\D\otimes\Ecal}(\alpha)$ factors
\begin{align}
\Q_{\D\otimes\Ecal}(\alpha) & =\sum_{\substack{\{x_{i}\}\\
i=1\dots r
}
}\tr\prod_{i=1}^{r}\D\otimes\Ecal_{x_{i}x_{\alpha(i)}}\nonumber \\
 & =\sum_{\substack{\{x_{i}',x_{i}''\}\\
i=1\dots r
}
}\tr\prod_{i=1}^{r}\D_{x_{i}'x_{\alpha(i)}'}\otimes\Ecal_{x_{i}''x_{\alpha(i)}''}\nonumber \\
 & =\sum_{\substack{\{x_{i}'\}\\
i=1\dots r
}
}\tr\prod_{i=1}^{r}\D_{x_{i}'x_{\alpha(i)}'}\sum_{\substack{\{x_{i}''\}\\
i=1\dots r
}
}\tr\prod_{i=1}^{r}\Ecal_{x_{i}''x_{\alpha(i)}''}\nonumber \\
 & =\Q_{\D}(\alpha)\Q_{\Ecal}(\alpha)\,.\label{eq:q-factoring}
\end{align}

\subsection{Average moments of $\An$\label{sub:Average-moments}}

If we set $\Ccal=\An$ the average moment factors as above, the dimension
is $k=d^{n}$ and according to \eqref{eq:q-factoring} we get
\begin{equation}
M_{r}(\A^{\otimes n})=C_{d^{n},\, r}\sum_{\alpha\in Sym(r)}\Q_{\A,\, r}(\alpha)^{n}\,.\label{eq:a-n-average-moments-evaluated}
\end{equation}
When the meaning is clear from the context we suppress the index in
$\Q_{\A}$ . The limiting behavior of this sum is relatively simple
and determines the quantity $\beta_{r}^{\reg}(\A)$ as described in
the following theorem.

\theoremstyle{plain}\newtheorem*{repeat_thm_br}{Theorem \ref{thm:beta-reg}}

\begin{repeat_thm_br}

\begin{itemize}\item[(a)]Let $r\geq2$ be an integer and let $\Q_{\max}=\max_{\alpha\in Sym(r)}|\Q_{{\cal A},\, r}(\alpha)|$.
Then 
\[
\beta_{r}^{\reg}(\A)=\frac{r\log d-\log\Q_{\max}}{r-1}\,.
\]

\item[(b)]Let $r\ge2$ be an integer. If the maximum $\Qmax$ is
attained for a unique $\alpha$ then the $\liminf$ in $\beta_{r}^{\reg}$
can be replaced with a regular limit:

\[
\beta_{r}^{\reg}(\A)=\lim_{n\to\infty}\frac{1}{n}\beta_{r}(\A^{\ot n})\,.
\]

\item[(c)]If the channel $\A$ is entrywise positive then for all
integer $r\geq2$
\[
\beta_{r}^{\reg}(\A)=\lim_{n\to\infty}\frac{1}{n}\beta_{r}(\A^{\ot n})\,.
\]

\end{itemize}

\end{repeat_thm_br}
\begin{proof}
(a) From the definition of \eqref{eq:def-regs} and \eqref{eq:a-n-average-moments-evaluated}
we have
\begin{align*}
\beta_{r}^{\reg}(\A) & =\liminf_{n\to\infty}\frac{1}{n(1-r)}\log M_{r}(\A^{\ot n})\\
 & =\liminf_{n\to\infty}\frac{1}{n(1-r)}\log\biggl[C_{d^{n},\, r}\sum_{\alpha}\Q(\alpha)^{n}\biggr]
\end{align*}
Note that
\begin{align*}
C_{d^{n},\, r}\sum_{\alpha}\Q(\alpha)^{n} & =\biggl|C_{d^{n},\, r}\sum_{\alpha}\Q(\alpha)^{n}\biggr|\\
 & \leq C_{d^{n},\, r}\sum_{\alpha}|\Q(\alpha)|^{n}\\
 & \leq d^{-nr}\sum_{\alpha}\Qmax^{n}\\
 & =r!\, d^{-nr}\Qmax^{n}
\end{align*}

Since $M_{r}(\A^{\ot n})\leq1$ we have
\begin{align*}
\beta_{r}^{\reg}(\A) & =\frac{1}{r-1}\liminf_{n\to\infty}\frac{1}{n}\log\biggl[C_{d^{n},\, r}\sum_{\alpha}\Q(\alpha)^{n}\biggr]^{-1}\\
 & \geq\frac{1}{r-1}\liminf_{n\to\infty}\frac{1}{n}\log\left(r!\, d^{-nr}\Qmax^{n}\right)^{-1}\\
 & =\frac{1}{r-1}\log\left(d^{r}\Qmax^{-1}\right)\\
 & =\frac{r\log d-\log\Qmax}{r-1}\,.
\end{align*}

In order to prove equality, we will use the existence of a subsequence
$\{n_{j}\}$ such that
\begin{equation}
\lim_{j\to\infty}\frac{1}{n_{j}(1-r)}\log\biggl[C_{d^{n_{j}},\, r}\sum_{\alpha}\Q(\alpha)^{n_{j}}\biggr]\leq\frac{r\log d-\log\Qmax}{r-1}\,.\label{eq:subsequence-inequality}
\end{equation}
To this end, let $\{\alpha_{1},\dots,\alpha_{N}\}$ be the maximizers
satisfying $|\Q(\alpha_{i})|=\Qmax$, so that $\Q(\alpha_{i})=\Qmax e^{2\pi i\gamma_{i}}$
for some $\gamma_{i}\in[0,1)$. It is a basic result from simultaneous
Diophantine approximations (using the Dirichlet box principle) that
for any $\epsilon>0$ there is an increasing sequence of positive
integers $n_{j}$ such that $ $$\max_{i}\{n_{j}\gamma_{i}\}<\epsilon$
for all $j$, where $\{x\}$ denotes the distance to the closest integer.
Choose $\epsilon=1/6$, then we have
\begin{align}
\bigg|\sum_{i=1}^{N}\Q(\alpha_{i})^{n_{j}}\bigg| & =\Q_{\max}^{n_{j}}\bigg|\sum_{i=1}^{N}e^{2\pi in_{j}\gamma_{i}}\bigg|\nonumber \\
 & \geq\Q_{\max}^{n_{j}}\cdot\frac{N}{2}\,.\label{eq:q-sum-lower-bound}
\end{align}
Furthermore, there is $\Theta<1$ such that
\[
|\Q(\alpha')|\le\Qmax\cdot\Theta
\]
for all $\alpha'$ which are not maximizers. Thus
\begin{align*}
C_{d^{n_{j}},\, r}\sum_{\alpha}\Q(\alpha)^{n_{j}} & =C_{d^{n_{j}},\, r}\biggl|\sum_{\alpha}\Q(\alpha)^{n_{j}}\biggr|\\
 & \ge C_{d^{n_{j}},\, r}\Qmax^{n_{j}}\cdot\frac{N}{2}-C_{d^{n_{j}},\, r}\biggl|\sum_{\alpha'}\Q(\alpha')^{n_{j}}\biggr|\\
 & \ge C_{d^{n_{j}},\, r}\Qmax^{n_{j}}\left(\frac{N}{2}-r!\,\Theta^{n_{j}}\right)
\end{align*}
Since $\Theta<1$, for $j$ sufficiently large we have $\frac{N}{2}-r!\Theta^{n_{j}}\geq\frac{1}{3}$,
thus for $j$ sufficiently large
\begin{align*}
\frac{1}{n_{j}(1-r)} & \biggl[\log C_{d^{n_{j}},\, r}\,\sum_{\alpha}\Q(\alpha)^{n_{j}}\biggr]\\
 & \le\frac{1}{r-1}\frac{1}{n_{j}}\log\left(3\Qmax^{-n_{j}}\, C_{d^{n_{j}},\, r}^{-1}\right)\,.
\end{align*}
Using $\lim_{n\to\infty}\left(d^{nr}C_{d^{n},\, r}\right)=1$, the
result follows immediately.\medskip{}

(b) Let $\alpha_{0}\in Sym(r)$ be the unique permutation satisfying
$|\Q(\alpha_{0})|=\Qmax$, so that $N=1$ in the notation of (a).
Then \eqref{eq:q-sum-lower-bound} is replaced by the equality
\[
\left|\Q(\alpha_{0})^{n}\right|=\Qmax^{n}
\]
which holds for every $n$. Thus the inequality \eqref{eq:subsequence-inequality}
is true for every $n$, hence the upper and lower bound yield the
existence of the limit.\medskip{}

(c) If $\A$ is entrywise positive then every term $\Q(\alpha)$ is
also positive. Thus the inequality \eqref{eq:q-sum-lower-bound} is
replaced by
\[
\biggl|\sum_{i=1}^{N}\Q(\alpha_{i})^{n}\biggr|=N\Qmax^{n}
\]
and this holds for every $n$. Thus again \eqref{eq:subsequence-inequality}
is true for every $n$, and the result follows.
\end{proof}

\subsection{Entanglement breaking channels\label{sub:Entanglement-breaking-channels}}

After dealing with basic facts about $\beta_{r}^{\reg}$ we turn our
attention to the special case of entanglement breaking channels. In
this case calculating the relevant $\Q$-terms and studying their
properties is particularly easy. We restate and then prove Theorem~\ref{thm:beta-reg-eb}
from the introduction.

\theoremstyle{plain}\newtheorem*{repeat_thm_eb}{Theorem \ref{thm:beta-reg-eb}}

\begin{repeat_thm_eb}

\begin{itemize} \item[(a)]Let $r\geq2$ be an integer, and $\A(\rho)=\sum\sigma_{k}\Tr{X_{k}\rho}$
an entanglement breaking channel satisfying the condition $\Tr{\prod_{i=1}^{r}\sigma_{k_{i}}}\geq0$
for all choices of $\{k_{i}\}$. Then

\[
\beta_{r}^{\reg}(\A)=\lim_{n\to\infty}\frac{1}{n}\beta_{r}(\A^{\ot n})=\frac{1}{1-r}\log\tr({\cal A}(\one/d)^{r})\,.
\]

\item[(b)]Let ${\cal A}$ be a d-dimensional unital entanglement-breaking
channel, $\A(\one)=\one$, satisfying the condition $\Tr{\prod_{i=1}^{m}\sigma_{k_{i}}}$
for all integer $m\geq2$. Then for all real $r\geq1$
\[
\Srreg(\A)=\beta_{r}^{\reg}(\A)=\log d\,.
\]

\end{itemize}

\end{repeat_thm_eb}
\begin{proof}
(a) The result follows immediately from Lemma~\ref{lem:ebc} and
Theorem~\ref{thm:beta-reg}.

(b) For a unital entanglement breaking channel $\beta_{r}^{\reg}(\A)$
attains its maximal value $\log d$ which is at the same time the
maximal possible value of $\Srreg(\A)$. So, with inequality~\eqref{eq:entropy-ineq}
it follows that $\Srreg(\A)=\beta_{r}^{\reg}(\A)=\log d$ for all
integer $r\geq2$. In the following we extend this result to all real
$r\geq1$.

Consider the derivatives of the function $f_{n}(r)=\frac{1}{n}\ln\E{\tr\left(\A^{\ot n}(\ket{\phi}\bra{\phi})^{r}\right)}/\ln2$
with respect to $r$, and set $\rho=\A^{\ot n}(\ket{\phi}\bra{\phi})$
\begin{align}
f_{n}'(r) & =\frac{1}{n}\E{\tr\rho^{r}}^{-1}\E{\Tr{\rho^{r}\ln\rho}}/\ln2<0\nonumber \\
f_{n}''(r) & =\frac{1}{n}\E{\tr\rho^{r}}^{-1}\E{\Tr{\rho^{r}(\ln\rho)^{2}}}/\ln2-\frac{1}{n}\E{\tr\rho^{r}}^{-2}\E{\Tr{\rho^{r}\ln\rho}}^{2}/\ln2>0\label{eq:second-inequality}
\end{align}
To prove the second inequality we use two applications of the Cauchy-Schwarz
inequality. First $|\Tr{AB}|\le\Tr{A^{2}}^{1/2}\Tr{B^{2}}^{1/2}$
with $A=\rho^{r/2}$ and $B=\rho^{r/2}\ln\rho$, and then $\E{XY}^{2}\le\E{X^{2}}\E{Y^{2}}$
to deduce

\[
\E{\Tr{\rho^{r}\ln\rho}}^{2}\le\E{(\Tr{\rho^{r}(\ln\rho)^{2}})^{1/2}\,(\tr\rho^{r})^{1/2}}^{2}\le\E{\Tr{\rho^{r}(\ln\rho)^{2}}}\,\E{\tr\rho^{r}}
\]
Therefore, the function $f_{n}(r)$ is convex in $r$ for $r>0$.
We also know that $f_{n}(r)\geq(1-r)\,\log d$ for any real $r\geq1$
and $\lim_{n\to\infty}f_{n}(r)=(1-r)\,\log d$ for any integer $r\geq1$.
Therefore, for any integer $r_{0}\geq1$ and $r\in[r_{0},r_{0}+1]$
we have upper and lower bounds
\begin{equation}
f_{n}(r_{0})(r_{0}+1-r)+f_{n}(r_{0}+1)(r-r_{0})\geq f_{n}(r)\geq(1-r)\,\log d\,.\label{eq:fn-bound}
\end{equation}
Thus we have $\lim_{n\to\infty}f_{n}(r)=(1-r)\,\log d$ for all real
$r\geq1$. Dividing by $1-r$ gives $\lim_{n\to\infty}\frac{1}{n}S_{r}(\A^{\ot n})=\log d$
the desired equality for all $r>1$.

Finally, for $r=1$ the Renyi entropy is defined as the Neumann entropy,
which equals the limit $\frac{1}{n}S_{1}(\A^{\ot n})=\lim_{r\to1}\frac{f_{n}(r)}{1-r}=-f_{n}'(1)$.
Again, using the bound \eqref{eq:fn-bound} with $r_{0}=1$ (and noting
that $f_{n}(1)=0$) we get $-f_{n}(2)\,(1-r)\geq f_{n}(r)\geq(1-r)\,\log d$
and so
\[
-f_{n}(2)\leq\frac{f_{n}(r)}{1-r}\leq\log d
\]
which implies $-f_{n}(2)\leq\frac{1}{n}S_{1}(\A^{\ot n})\leq\log d$.
But as $\lim_{n\to\infty}-f_{n}(2)=\log d$ both sides of the bound
become equal and we have $\lim_{n\to\infty}\frac{1}{n}S_{1}(\A^{\ot n})=\log d$
as well.\end{proof}
\begin{lem}
\label{lem:ebc}Let $\A$ be an entanglement breaking channel with
\begin{equation}
\Tr{\prod_{i=1}^{r}\sigma_{k_{i}}}\geq0\label{eq:ebc-condition}
\end{equation}
 for all choices of $\{k_{i}\}$, then $\alpha=\id$ is the unique
maximum of $Q_{\A,\, r}$ so $\Qmax=\Q_{\A,\, r}(\id)=\Tr{\A(\one)^{r}}$.\end{lem}
\begin{proof}
Entanglement breaking channels are of the form $\A(\rho)=\sum\sigma_{k}\Tr{X_{k}\rho}$
where the $\sigma_{k}$ are density matrices and the $X_{k}$ constitute
a POVM. Therefore we calculate
\begin{align}
|\Q(\alpha)| & =\bigg|\sum_{\substack{\{x_{i}\}\\
i=1\dots r
}
}\tr\prod_{i=1}^{r}\A_{x_{i}x_{\alpha(i)}}\bigg|\nonumber \\
 & =\bigg|\sum_{\substack{\{x_{i}\}\\
i=1\dots r
}
}\Tr{\prod_{i=1}^{r}\sum_{k}\sigma_{k}\Tr{X_{k}\ket{x_{\vphantom{()}i}}\bra{x_{\alpha(i)}}}}\bigg|\nonumber \\
 & =\bigg|\sum_{\substack{\{x_{i},k_{i}\}\\
i=1\dots r
}
}\Tr{\prod_{i=1}^{r}\sigma_{k_{i}}}\prod_{i=1}^{r}\Tr{X_{k_{i}}\ket{x_{\vphantom{()}i}}\bra{x_{\alpha(i)}}}\bigg|\nonumber \\
 & \leq\sum_{\substack{\{k_{i}\}\\
i=1\dots r
}
}\Tr{\prod_{i=1}^{r}\sigma_{k_{i}}}\bigg|\sum_{\substack{\{x_{i}\}\\
i=1\dots r
}
}\prod_{i=1}^{r}\bra{x_{\alpha(i)}}X_{k_{i}}\ket{x_{\vphantom{(}i}}\bigg|\label{eq:ebc-qal-bound}
\end{align}
where condition~\eqref{eq:ebc-condition} is used in the last step.
The term on the right side of the last line can be rewritten as a
product of traces
\begin{align*}
\sum_{\substack{\{x_{i}\}\\
i=1\dots r
}
}\prod_{i=1}^{r}\bra{x_{\alpha(i)}}X_{k_{i}}\ket{x_{\vphantom{(}i}} & =\sum_{\substack{\{x_{i}\}\\
i=1\dots r
}
}\dots\bra{x_{\alpha^{2}(1)}}X_{k_{\alpha(1)}}\ket{x_{\alpha(1)}}\bra{x_{\alpha(1)}}X_{k_{1}}\ket{x_{\vphantom{(}1}}\\
 & =\prod_{\gamma\in\alpha}\tr\prod_{i\in\gamma^{-1}}X_{k_{i}}
\end{align*}
where $\gamma\in\alpha$ are the sub-cycles of $\alpha$ and $\prod_{i\in\gamma^{-1}}$
is a product over the numbers in $\gamma^{-1}$.

Now consider any set of operators $Y_{j}\geq0$, using the spectral
decomposition $Y_{j}=\sum_{k}\lambda_{k,\, j}\ket k_{j}\bra k_{j}$
we have 
\begin{align}
\bigg|\tr\prod_{j=1}^{m}Y_{j}\bigg| & =\bigg|\sum_{\substack{\{k_{j}\}\\
j=1\dots m
}
}\lambda_{k_{1},\,1}\dots\lambda_{k_{m},\, m}\Tr{\ket{k_{1}}_{1}\bra{k_{1}}_{1}\dots\ket{k_{m}}_{m}\bra{k_{m}}_{m}}\bigg|\label{eq:first-line}\\
 & \leq\sum_{\substack{\{k_{j}\}\\
j=1\dots m
}
}\lambda_{k_{1},\,1}\dots\lambda_{k_{m},\, m}=\prod_{j=1}^{m}\tr Y_{j}\,.\nonumber 
\end{align}
Equality holds only in the following cases:
\begin{itemize}
\item If $m=1$.
\item If any of the $Y_{j}$'s equals zero.
\item If all the $Y_{j}$ are a multiple of a one-dimensional projection.
\end{itemize}
To see that there are no other possibilities consider the case where
$m\geq2$ and all $Y_{j}$ are rank one but they don't have the same
eigenvectors. Now the sum in \eqref{eq:first-line} contains the overlap
of the eigenvectors which is smaller than one in absolute value because
some eigenvectors are not the same. Therefore there is no equality.
Finally, consider the case where $m\geq2$, where all the $Y_{j}$
have at least rank one and where there exists a $j_{0}$ such that
$Y_{j_{0}}$ has rank two or higher. On the RHS of \eqref{eq:first-line},
choose the $k_{j}$ such that all $\lambda_{k_{j}}$ are non-zero.
For $k_{j_{0}}$ there are two or more choices and for one of those
choices the trace term has to be smaller than one in absolute value.
Therefore equality cannot hold in this case.

Returning to the $X_{j}$ we see that if any $X_{j}=0$ we can drop
it from our POVM without changing the channel $\A$. Also if say $X_{1}=qX_{2}$
are multiples of each other then we can combine them to $\tilde{X}_{1}=X_{1}+X_{2}$
and $\tilde{\sigma}_{1}=(\sigma_{1}+q\sigma_{2})/(1+q)$ again without
changing $\A$. Therefore we can assume no $X_{j}$ equals zero and
no two $X_{j}$ are multiples of each other, then equality is only
possible if $m=1$. It follows that
\begin{align*}
\bigg|\prod_{\gamma\in\alpha}\tr\prod_{i\in\gamma^{-1}}X_{k_{i}}\bigg| & \leq\prod_{i=1}^{r}\tr X_{k_{i}}\,,
\end{align*}
can only be equality if all cycles in $\alpha$ have length one, i.e.
$\alpha=\id$.

Combining the last inequality with \eqref{eq:ebc-qal-bound} we get

\[
|\Q(\alpha)|\leq\sum_{\substack{\{k_{i}\}\\
i=1\dots r
}
}\Tr{\prod_{i=1}^{r}\sigma_{k_{i}}}\prod_{i=1}^{r}\tr X_{k_{i}}=\Q(\id)
\]
with equality if and only if $\alpha=\id$. And finally,
\begin{align*}
\Q(\id) & =\sum_{\substack{\{k_{i}\}\\
i=1\dots r
}
}\Tr{\prod_{i=1}^{r}\sigma_{k_{i}}}\prod_{i=1}^{r}\tr X_{k_{i}}\\
 & =\sum_{\substack{\{k_{i}\}\\
i=1\dots r
}
}\Tr{\prod_{i=1}^{r}\sigma_{k_{i}}\tr X_{k_{i}}}\\
 & =\Tr{\prod_{i=1}^{r}\sum_{k}\sigma_{k}\tr X_{k}}\\
 & =\Tr{\A(\one)^{r}}\,.
\end{align*}

\end{proof}

\subsection{The qubit depolarizing channel}

\subsubsection{\label{sub:evaluating-q}Evaluating $\Q_{\Delta_{\lambda}(\alpha)}$}

In this section we calculate $\Q_{\Delta_{\lambda}}(\alpha)$ for
$\alpha=\id$ and $\alpha=(1\dots r)$. As we prove in Lemma~\ref{pro:maximal-q}
in \ref{sub:maximal-q} one of these two terms is always maximal,
$\Q_{\max}=\max\left\{ \Q_{\Delta_{\lambda}}(\id),\Q_{\Delta_{\lambda}}((1\dots r))\right\} $,
so they are of particular interest.

We have
\[
\Q_{\Delta_{\lambda}}(\id)=\tr\Delta_{\lambda}(\one)^{r}=\tr\one=2
\]

To evaluate the second $\Q$-term we consider a slightly more general
channel $\A$ with Choi-Jamiolkowski representation
\[
Choi(\Delta_{\lambda})=\begin{pmatrix}\mu & 0 & 0 & \lambda\\
0 & \nu & \kappa & 0\\
0 & \kappa & \nu & 0\\
\lambda & 0 & 0 & \mu
\end{pmatrix}\,.
\]
This matrix has diagonal block form with blocks
\[
\begin{pmatrix}\mu & \lambda\\
\lambda & \mu
\end{pmatrix},\,\begin{pmatrix}\nu & \kappa\\
\kappa & \nu
\end{pmatrix}\,.
\]
We need to raise the matrix to the $r$-th power, which gives
\[
\Mtwo{\mu}{\lambda}{\lambda}{\mu}^{r}=\frac{1}{2}\Mtwo 111{-1}\Mtwo{(\mu+\lambda)^{r}}00{(\mu-\lambda)^{r}}\Mtwo 111{-1}
\]
and similarly for the second matrix. Therefore we get
\begin{equation}
\Q_{\A}((1\dots r))=\Tr{Choi_{\Delta_{\lambda}}^{r}}=(\mu+\lambda)^{r}+(\mu-\lambda)^{r}+(\nu+\kappa)^{r}+(\nu-\kappa)^{r}\,.\label{eq:q-1tor-dualrail}
\end{equation}
For $\A=\Delta_{\lambda}$ we have $\mu=\frac{1+\lambda}{2}$, $\nu=\frac{1-\lambda}{2}$
and $\kappa=0$, so
\[
\Q_{\Delta_{\lambda}}((1\dots r))=\fracp{1+3\lambda}2^{r}+3\fracp{1-\lambda}2^{r}\,.
\]
In \ref{sec:q-dim-d} we compute the $\Q((1\dots r))$ for any dimension
$d\geq2$.

More generally, whenever $\alpha$ is a product of cycles of consecutive
numbers the sum factors as shown in Lemma~\ref{lem:q-factoring},
e.g. $\Q((123)(45))=\frac{1}{d}\Q((123))\Q((45))$.

\subsubsection{Regularized output entropy of $\Deltaln$\label{sub:average-output-entropy}}

\begin{figure}
\begin{centering}
\includegraphics[height=7cm]{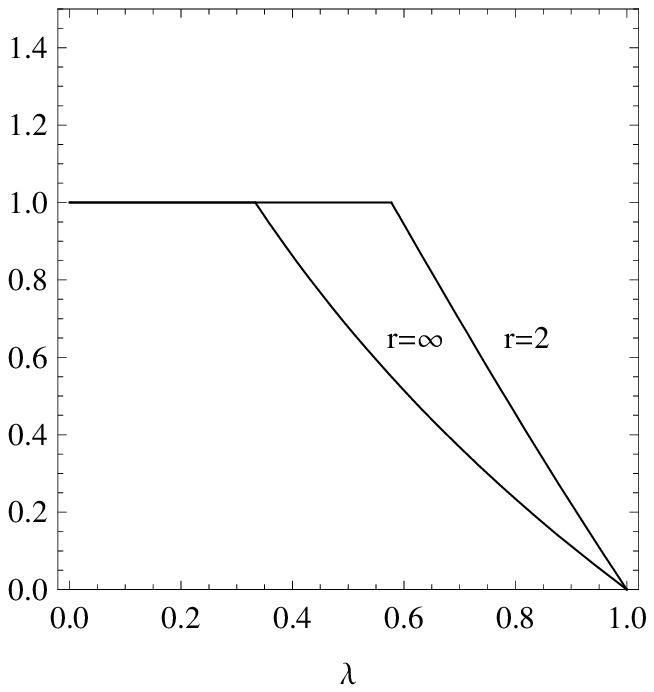}
\par\end{centering}

\caption{$S_{2}$ and $S_{\infty}$.}

\label{fig:entropy-plot} 
\end{figure}
\begin{table}
\begin{centering}
\begin{tabular}{|c||c|c|}
\hline 
$r$ & $c_{r}$ & $d_{r}$\tabularnewline
\hline 
\hline 
2 & .577 & .732\tabularnewline
\hline 
3 & .5 & .835\tabularnewline
\hline 
4 & .458 & .878\tabularnewline
\hline 
10 & .381 & .953\tabularnewline
\hline 
100 & .338 & .995\tabularnewline
\hline 
\end{tabular}
\par\end{centering}

\caption{Range parameters of Theorem~\ref{thm:main}.\label{tab:validity}}
\end{table}
We are now ready to prove the main theorem.

\theoremstyle{plain}\newtheorem*{repeat_thm_main}{Theorem~\ref{thm:main}}

\begin{repeat_thm_main}

\begin{itemize}\item[(a)]For all $r\in\mathbb{N}$, $r\ge2$, and
$\lambda\in[0,1]$, 
\begin{align}
\beta_{r}^{\reg}(\Delta_{\lambda}) & =\lim_{n\rightarrow\infty}\,\frac{1}{n}\,\beta_{r}(\Deltaln)=\min\left\{ 1,\frac{2r-\log\left[(1+3\lambda)^{r}+3(1-\lambda)^{r}\right]}{r-1}\right\} \label{eq:beta-reg-depol}
\end{align}
in particular
\begin{align*}
\beta_{2}^{\reg}(\Delta_{\lambda}) & =\begin{cases}
1 & \lambda\leq1/\sqrt{3}\\
2-\log(1+3\lambda^{2}) & \lambda>1/\sqrt{3}
\end{cases}\,,\\
\beta_{\infty}^{\reg}(\Delta_{\lambda}) & =\begin{cases}
1 & \lambda\leq1/3\\
2-\log(1+3\lambda) & \lambda>1/3
\end{cases}\,,
\end{align*}
where $\beta_{\infty}^{\reg}(\Delta_{\lambda})=\lim_{r\to\infty}\,\beta_{r}^{\reg}(\Delta_{\lambda})$.

\item[(b)]For all $r\in\mathbb{N}$, $r\ge2$, and $\lambda\in J_{r}$,
\begin{equation}
\overline{S}_{r}^{\reg}(\Delta_{\lambda})=\beta_{r}^{\reg}(\Delta_{\lambda})\label{eq:s-avg-equality}
\end{equation}
where $J_{r}=[0,c_{r}]\cup[d_{r},1]\subset[0,1]$ for some $0<c_{r}<d_{r}<1$
(see Table~\ref{tab:validity}).

\end{itemize}

\end{repeat_thm_main}

The regularized output entropy for $r=2$ and the lower bound of the
same for $r=\infty$ are plotted in Figure~\ref{fig:entropy-plot}.
We expect that \eqref{eq:s-avg-equality} holds for all $\lambda$.
The regularized output entropy is maximal when $2\geq\fracp{1+3\lambda}2^{r}+3\fracp{1-\lambda}2^{r}$
which holds for all $r$ when $\lambda\leq1/3$. It's interesting
to note that $\lambda\leq1/3$ is also the condition for $\Delta_{\lambda}$
to be entanglement breaking.
\begin{proof}
(a) From \ref{sec:q-dim-d} we know that $\Q(\id)=2$ and $\Q((1\dots r))=\fracp{1+3\lambda}2^{r}+3\fracp{1-\lambda}2^{r}$.
In Lemma~\ref{pro:maximal-q} in \ref{sub:maximal-q} we prove that
one of these two $\Q$-terms yields the maximal value, i.e. for any
$\lambda\in[0,1]$
\[
\Qmax=\max\left\{ 2,\fracp{1+3\lambda}2^{r}+3\fracp{1-\lambda}2^{r}\right\} \,.
\]
Additionally, $\Delta_{\lambda}$ is entrywise positive for $\lambda\in[0,1]$.
Therefore we can apply parts (a) and (c) of Theorem~\ref{thm:beta-reg}
and \eqref{eq:beta-reg-depol} follows immediately.

In the case $r=2$ we have
\begin{align*}
\frac{2\cdot2-\log\left[(1+3\lambda)^{2}+3(1-\lambda)^{2}\right]}{2-1} & =4-\log(4+12\lambda^{2})\\
 & =2-\log(1+3\lambda^{2})\,,
\end{align*}
and in the case $r=\infty$ we have
\begin{align*}
\lim_{r\to\infty}\,\beta_{r}^{\reg}(\Delta_{\lambda}) & =\lim_{r\to\infty}\,\frac{2r-\log\left[(1+3\lambda)^{r}+3(1-\lambda)^{r}\right]}{r-1}\\
 & =\lim_{r\to\infty}\,\frac{2r-r\,\log(1+3\lambda)}{r-1}\\
 & =2-\log(1+3\lambda)\,.
\end{align*}

(b) Define the functions
\begin{align*}
f(\ket{\phi}) & =\Tr{\Deltaln(\ket{\phi}\bra{\phi})^{r}}\\
g(x) & =\frac{1}{n(1-r)}\log x\,,
\end{align*}
and the two series (notice the $n$ dependency in the definition of
$f$)

\begin{align*}
a_{n} & =g(\e f)\,,\\
b_{n} & =\e g\circ f\,.
\end{align*}
The only difference between $a_{n}$ and $b_{n}$ is the position
at which the averaging over pure inputs $\ket{\phi}$ takes places.
As a result, $a_{n}=\frac{1}{n}\beta_{r}(\Deltaln)$ contains the
average moments that have been considered in previous sections, and
$b_{n}=\frac{1}{n}S_{r}(\Deltaln)$ is the Renyi output entropy per
systems. By Jensen's inequality and because the maximal output entropy
is $\log2=1$ we have the bounds $a_{n}\leq b_{n}\leq1$ for any $\lambda$.
Our goal is to prove that the two series have the same limit for $\lambda\in[0,c_{r}]\cup[d_{r},1]$.

From part (a) we know
\[
\lim_{n\to\infty}a_{n}=\frac{r-\log\Qmax}{r-1}\,.
\]
First, choose $c(r)$ such that $2\geq\fracp{1+3\lambda}2^{r}+3\fracp{1-\lambda}2^{r}$
for all $\lambda\in[0,c]$. Then we have $\Qmax=2$, and so 
\[
\lim_{n\to\infty}a_{n}=\lim_{n\to\infty}b_{n}=\Srreg(\Delta_{\lambda})=1\,.
\]
Now only \eqref{eq:s-avg-equality} remains to be proved for the range
$\lambda\in[d,1]$ with $d$ still to be determined.

In Proposition~\ref{pro:lipschitz-constant} in \ref{sub:lipschitz-constant}
we prove that the Lipschitz constant $\eta$ of the map $f$ is bounded
by $\eta\leq\sqrt{2}r\kappa^{n}$ for $\kappa=\lambda+\frac{1-\lambda}{\sqrt{2}}$.
From Levy's Lemma (according to Lemma III.1 in \cite{hayden_leung_winter_aspectsofentanglement})
we know that the values of $f$ concentrate around their average 
\begin{equation}
\Pr(|f(\ket{\phi})-\e f|>\alpha_{n})\leq4\exp\left(-C(k+1)\frac{\alpha_{n}^{2}}{\eta^{2}}\right)=:\epsilon_{n}\label{eq:pr-f-bound}
\end{equation}
with $k=2\cdot2^{n}-1$ the (real) dimension of the sphere of input
states, $C=(9\pi^{3}\ln2)^{-1}$, and we choose the deviation
\[
\alpha_{n}=\frac{1}{2}N\fracp{\Qmax}{2^{r}}^{n}\approx\frac{1}{2}\e f
\]
where $N\geq1$ is the multiplicity of the maximum of $\Q(\alpha)$,
that is, $\alpha_{n}$ is half of the dominant term of $\e f$ (see
prove of Theorem~\ref{thm:beta-reg}). Now $\alpha_{n}\to0$ (apart
from the special case $\lambda=1$) $ $and $\e f-\alpha_{n}>0$ which
is required in a later step. To ensure concentration for large $n$
the exponent $(k+1)\frac{\alpha_{n}^{2}}{\eta^{2}}$ needs to become
large. Taking the $2n$-th root we get
\begin{align*}
\left((k+1)\frac{\alpha_{n}^{2}}{\eta^{2}}\right)^{1/2n} & =\left(2\cdot\sqrt{2}^{2n}\frac{\frac{N}{4}\left(\frac{\Qmax}{2^{r}}\right)^{2n}}{2r^{2}\kappa^{2n}}\right)^{1/2n}\\
 & \approx\sqrt{2}\frac{\frac{\Qmax}{2^{r}}}{\kappa}
\end{align*}
where the approximation becomes equality in the large $n$ limit.
If this term is larger than 1 we have the required divergence, this
gives the following inequality 
\begin{align}
\frac{\Qmax}{2^{r}} & >\frac{\kappa}{\sqrt{2}}\,.\label{eq:range-of-validity}
\end{align}
This condition gives the lower bound $d_{r}$ for the $\lambda$ values
for which \eqref{eq:s-avg-equality} holds. Both sides of the inequality
are plotted in Figure~\ref{fig:validity}. For some values of $r$
the range parameters can be found in Table~\ref{tab:validity}.
\begin{figure}
\begin{centering}
\includegraphics[height=7cm]{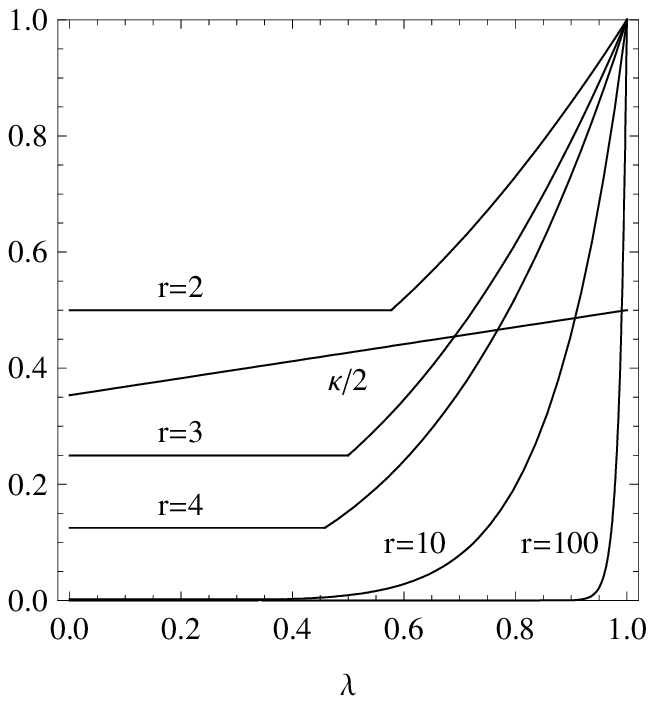}
\par\end{centering}

\caption{Plots of both sides of \eqref{eq:range-of-validity} for some $r$
values to determine the range of validity.\label{fig:validity}}

\end{figure}

To transform \eqref{eq:pr-f-bound} into a statement about the range
of $g$ -- and whence about $a_{n}$ -- we need a condition that implies
$|f(\ket{\phi})-\e f|>\alpha_{n}$. We set $\alpha_{n}'=\frac{1}{n(r-1)}\frac{\alpha_{n}}{\e f-\alpha_{n}}>|g(\e f)-g(\e f-\alpha_{n})|>|g(\e f)-g(\e f+\alpha_{n})|$
where the second inequality follows from convexity of $g$. Then for
any $x$ we have
\[
|g(\e f)-g(x)|>\alpha_{n}'\Rightarrow|\e f-x|>\alpha_{n}\,.
\]
Therefore, the values of $g\circ f$ concentrate around $a_{n}=g(\e f)$
\[
\Pr(|g\circ f(\ket{\phi})-a_{n}|>\alpha_{n}')\leq\epsilon_{n}\,.
\]
For $\alpha_{n}'$ we calculate
\begin{align*}
\lim_{n\to\infty}\alpha_{n}' & =\lim_{n\to\infty}\frac{1}{n(r-1)}\frac{\alpha_{n}}{\e f-\alpha_{n}}\\
 & =\lim_{n\to\infty}\frac{1}{n(r-1)}\frac{1/2}{1/2}\\
 & =0\,.
\end{align*}

To find an upper bound on $b_{n}$ assume all $\ket{\phi}\in(g\circ f)^{-1}[a_{n}-\alpha_{n}',a_{n}+\alpha_{n}']=B$
map to the maximal value $a_{n}+\alpha_{n}'$ and all $\ket{\phi}\in B^{c}$
map to the maximal value 1. This gives a possible range for the average
\begin{equation}
b_{n}\in[a_{n},(1-\epsilon_{n})(a_{n}+\alpha_{n}')+\epsilon_{n}\cdot1]\,.\label{eq:bn-range}
\end{equation}
For $\lambda\in[d,1]$ both $\epsilon_{n}$ and $\alpha'_{n}$ tend
to 0 for large $n$ and with \eqref{eq:bn-range} and because we know
the limit of $a_{n}$ we have
\[
\lim_{n\to\infty}a_{n}=\lim_{n\to\infty}b_{n}=\frac{r-\log\Qmax}{r-1}\,.
\]

\end{proof}

\subsubsection{Output entropy of random sequences}

Considering sequences of random pure input states $\ket{\phi_{n}}$
with increasing dimension $2^{n}$ we have the following statement.
\begin{prop}
Let $\ket{\phi_{n}}\in\C^{2^{n}}$ be a sequence of random pure states
and
\[
c_{n}=\frac{1}{n}S_{r}(\Deltaln(\ket{\phi_{n}}\bra{\phi_{n}}))
\]
the sequence of output Renyi entropies per system. Then
\[
c_{n}\stackrel{\mathrm{a.s.}}{\to}\frac{r-\log\Qmax}{r-1}
\]
 if $\lambda$ is restricted as in Theorem~\ref{thm:main}.\end{prop}
\begin{proof}
Define $f_{n}(\ket{\phi})=\frac{1}{n}S_{r}(\Deltaln(\ket{\phi}\bra{\phi}))$,
the limit value $c=\frac{r-\log\Qmax}{r-1}=\lim_{n\to\infty}\e f_{n}$,
and $\delta_{n}=|\e f_{n}-c|$. Then for any $\epsilon>0$
\begin{align}
\Pr(|f_{n}(\ket{\phi})-c|>\epsilon) & \leq\Pr(|f_{n}(\ket{\phi})-\e f_{n}|>\epsilon-\delta_{n})\label{eq:almost-prob}\\
 & \leq4\exp\left(-C(k+1)\frac{(\epsilon-\delta_{n})^{2}}{\eta^{2}}\right)\,,\nonumber 
\end{align}
with $C$, $k=2\cdot2^{n}-1$, and $\eta\leq\sqrt{2}r\kappa^{n}$
with $0<\kappa<1$ as in the proof of Theorem~\ref{thm:main}. If
$\lambda\in J_{r}$ then $\lim\delta_{n}=0$. When $n$ becomes large
then $k$ becomes large, $\eta$ becomes small, and $\epsilon-\delta_{n}$
is close to $\epsilon>0$. Therefore the probabilities in \eqref{eq:almost-prob}
become small. Let $N$ be such that for $n>N$ we have $\delta_{n}<\epsilon/2$.
Then we have the bound
\begin{align*}
\sum_{n=1}^{\infty}\Pr(|f_{n}(\ket{\phi_{n}})-c|>\epsilon) & \leq\sum_{n=1}^{\infty}4\exp\left(-C(k+1)\frac{(\epsilon-\delta_{n})^{2}}{\eta^{2}}\right)\\
 & \leq N+\sum_{n>N}4\exp\left(-C(k+1)\frac{(\epsilon/2)^{2}}{\eta^{2}}\right)\\
 & \leq N+\sum_{n>N}4\exp\left(-\tilde{C}(2/\kappa^{2})^{n}\right)\\
 & <\infty\,,
\end{align*}
with $\tilde{C}=C\cdot2\cdot\frac{(\epsilon/2)^{2}}{(\sqrt{2}r)^{2}}$
independent of $n$. The Proposition follows by the lemma of Borel-Cantelli.
\end{proof}

\section{Proof of lemmas}

\subsection{Maximal $\Q$ for $\Delta_{\lambda}$\label{sub:maximal-q}}

Before proving Proposition~\ref{pro:maximal-q} which is required
for the proof of Theorem~\ref{thm:main} we introduce some new notation
and we present four lemmas. We only prove our Lemmas in two dimensions,
but similar results will hold for higher dimensions.

For the most part we consider a channel $\A$ slightly more general
than the depolarization channel with
\[
\begin{array}{r@{=}lr@{=}l}
\A_{00} & \Mtwo{\mu}00{\nu} & \A_{11} & \Mtwo{\nu}00{\mu}\vspace{1em}\\
\A_{10} & \Mtwo 0{\kappa}{\lambda}0 & \A_{01} & \Mtwo 0{\lambda}{\kappa}0
\end{array}\,,
\]
where $\kappa,\lambda,\mu,\nu\in\R_{0}^{+}$ and $\mu\geq\nu$, $\lambda\geq\kappa$.
We call this channel the \emph{two-rail channel}. Remember that in
the case $\A=\Delta_{\lambda}$ we have $\mu=\frac{1+\lambda}{2}$,
$\nu=\frac{1-\lambda}{2}$ and $\kappa=0$. Because all these matrix
entries are positive, so are the $\Q$-terms, and therefore the largest
positive $\Q$-term will yield $\Qmax$.

We think of these matrices as the diagrams in Figure~\ref{fig:matrix-diagrams}.
We refer to the lines as rails and to their vertical position as their
track (starting with track~0) 
\begin{figure}
\begin{centering}
\includegraphics[width=3cm]{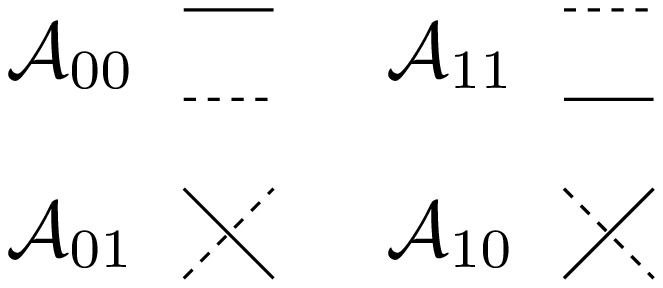}
\par\end{centering}

\caption{The diagrams for the $\A_{xy}$ matrices.}

\label{fig:matrix-diagrams} 
\[
\A_{00}\A_{01}\A_{10}\A_{11}\A_{10}
\]
\vspace{.3em}

\begin{centering}
\includegraphics[width=2.2cm]{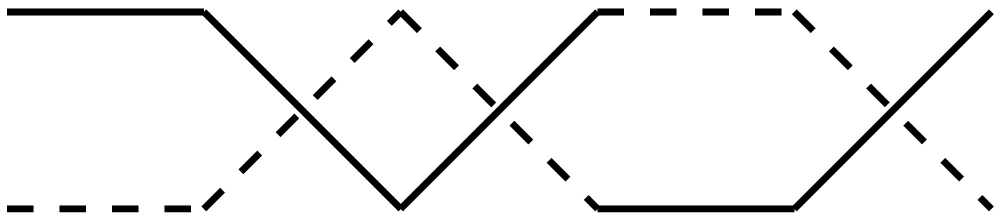}
\par\end{centering}

\caption{A product of $\A_{xy}$ matrices and the corresponding diagram.}

\label{fig:product-diagram} 
\end{figure}

A product looks like the diagram in Figure~\ref{fig:product-diagram},
notice that we read from right to left, the same way that matrix multiplication
applies. Consider a vector multiplying with this product, the diagram
can be thought of as presenting two rails along which the two entries
of the vector pass through to the left. On the way, continuous lines
multiply the entries with factors $\mu$ or $\lambda$, dashed lines
with a factor $\nu$ or $\kappa$. Because the products are inside
a trace, they will only contribute, if the rail starting at the top
on the right, ends on the top at the left (giving the $\bra 0\dots\ket 0$
contribution), and the same for the rail starting at the bottom (giving
the $\bra 1\dots\ket 1$ contribution). Using this, we can compare
contributions to $\Q(\alpha)$ for different $\alpha$.

We rewrite 
\[
\Q(\alpha)=\sum_{\delta\in\D(\alpha)}\Tr{\delta}
\]
 where $\D(\alpha)$ is the set of $2^{r}$ diagrams corresponding
to $\alpha$ and we identify the diagram $\delta$ with the corresponding
matrix.

Furthermore, for a diagram $\delta$ we define $\delta^{1}$ to be
the unchanged diagram and $\delta^{-1}$ to be the horizontally reflected
diagram. If we take $\D_{0}(\alpha)$ to be the set of all diagrams
in $\D(\alpha)$ that start high (at track~0) then obviously $\D(\alpha)=\D_{0}(\alpha)\cup\D_{0}^{-1}(\alpha)$
provides a convenient splitting of the sum in $\Q(\alpha)$.

We say an $\alpha$ is \emph{non-overlapping} if all its cycles permute
consecutive numbers, e.g. $\alpha=(123)(45)(678)$. For such $\alpha$
we write $\alpha=\alpha_{1}\dots\alpha_{s}$, where all the $\alpha_{i}$
are the cycles. We define a product for diagrams by simply concatenating
them. With this product we get $\D(\alpha)=\D(\alpha_{1})\dots\D(\alpha_{s})$.
\begin{lem}
\label{lem:prop-unity}Consider a diagram $\delta\in\D((1\dots r))$
of a two-rail channel. Then we have 
\[
\delta+\delta^{-1}\propto\one\,.
\]
\end{lem}
\begin{proof}
In any diagram $\delta\in\D((1\dots r))$ one rail starts and ends
at track~0 and the other starts and ends at track~1. That means
the corresponding matrix is diagonal. Reflecting the diagram simply
means exchanging the two diagonal entries, and if we sum $\delta+\delta^{-1}$
then the diagonal entries both have the same sum, i.e. it is proportional
to unity.\end{proof}
\begin{lem}
\label{lem:q-factoring}Let $\alpha=\alpha_{1}\dots\alpha_{s}\in Sym(r)$
be a non-overlapping permutation consisting of $s$ cycles. Then the
$\Q(\alpha)$ of a two-rail channel factors like

\[
\Q(\alpha)=\frac{1}{2^{s-1}}\Q(\alpha_{1})\dots\Q(\alpha_{s})\,.
\]
\end{lem}
\begin{proof}
First we use the fact, that the diagrams for a non-overlapping $\alpha$
can be split between cycles, i.e. the sum splits into sums over separate
diagrams,
\begin{align*}
\Q(\alpha) & =\sum_{\delta\in\D(\alpha)}\Tr{\delta}\\
 & =\sum_{\delta_{1}\in\D(\alpha_{1})}\dots\sum_{\delta_{s}\in\D(\alpha_{s})}\Tr{\delta_{1}\dots\delta_{s}}\\
 & =\Tr{\sum_{\delta_{1}\in\D(\alpha_{1})}\delta_{1}\dots\sum_{\delta_{s}\in\D(\alpha_{s})}\delta_{s}}\,.
\end{align*}
 Using the splitting $\D(\alpha)=\D_{0}(\alpha)\cup\D_{0}^{-1}(\alpha)$
and Lemma~\ref{lem:prop-unity} we have 
\[
\sum_{\delta\in\D(\alpha_{i})}\delta=\sum_{\delta\in\D_{0}(\alpha_{i})}\delta+\delta^{-1}\propto\one
\]
for all $s$ cycles $\alpha_{i}$. Therefore, we can split the single
trace into $s$ separate traces 
\begin{align*}
\Q(\alpha) & =\frac{1}{2^{s-1}}\sum_{\delta_{1}\in\D(\alpha_{1})}\Tr{\delta_{1}}\dots\sum_{\delta_{1}\in\D(\alpha_{s})}\Tr{\delta_{s}}\\
 & =\frac{1}{2^{s-1}}\Q(\alpha_{1})\dots\Q(\alpha_{s})\,.
\end{align*}
\end{proof}
\begin{lem}
For $\A=\Delta_{\lambda}$ the depolarizing channel in two dimensions,
$\Q$ restricted to non-overlapping permutations $\alpha\in Sym(r)$
is either maximal when $\alpha=(1\dots r)$ or when $\alpha=\id$.\label{lem:non-overlap-max}\end{lem}
\begin{proof}
Remember from \ref{sub:evaluating-q} that $\Q((1\dots r))=(\mu+\lambda)^{r}+3\nu^{r}$.
It is convenient to introduce
\[
f(x)=(\mu+\lambda)^{x}+3\nu^{x}\,,
\]
as a function with range $\R$. For a pure $\alpha=\alpha_{1}\dots\alpha_{s}$
according to Lemma~\ref{lem:q-factoring} we get 
\begin{equation}
\Q(\alpha)=\frac{1}{2^{s-1}}\prod_{i=1}^{s}f(|\alpha_{i}|)\,.\label{eq:q-for-pure-alpha}
\end{equation}
with $|\alpha_{i}|$ denoting the length of a cycle.

In the following we keep the number $s$ of cycles and the total length
$r$ of the permutation invariant. Now, if we increase the length
of one cycle and decrease the length of another, we are only changing
two factors in the product, $f(x)f(k-x)$, where $x$ is the length
of the first cycle and $k$ the (invariant) sum of the lengths of
both cycles. We rewrite
\begin{align*}
f(x)f(k-x) & =((\mu+\lambda)^{x}+3\nu^{x})((\mu+\lambda)^{k-x}+3\nu^{k-x})\\
 & =(\mu+\lambda)^{k}+9\nu^{k}+3\nu^{k}\fracp{\mu+\lambda}{\nu}^{x}+3(\mu+\lambda)^{k}\fracp{\nu}{\mu+\lambda}^{x}\,.
\end{align*}
Because of $\mu+\lambda,\nu\geq0$ the function $f(x)f(k-x)$ is convex
in $x$. Therefore it is maximal at the boundaries, i.e. when one
cycle has the minimal length of 1. If we repeat this procedure $s-1$
times we end up with one large cycle of length $t=r-s+1$ while all
other cycles are of length 1. In every step we increase $\Q$, so
we get the bound
\begin{equation}
\Q(\alpha)\leq\Q((1)(2)\dots(s-1)(s\dots r))=\frac{1}{2^{s-1}}f(1)^{s-1}f(t)=f(t)\label{eq:Q-bound}
\end{equation}
for non-overlapping permutations $\alpha$ consisting of $s$ cycles.

Now compare the upper bounds given by \eqref{eq:Q-bound} for permutations
of the same total length $r$ but different number of cycles $s$.
This is the same as varying $t$ . Because $f$ is convex we get a
maximal upper bound if $t$ is minimal or maximal. The minimal value
$t=1$ is achieved when $s=r$ and all the cycles are of length 1.
Then \eqref{eq:Q-bound} becomes an equality - there are no steps
necessary in the maximization procedure - and we have $\Q(\id)=f(1)=2$.
The maximal value $t=r$ is achieved when $\alpha$ is simply one
large cycle. Again \eqref{eq:Q-bound} becomes equality, and $\Q(\alpha)=f(r)$.
One of these upper bounds is the highest upper bound possible in \eqref{eq:Q-bound},
and because they are achieved by $\Q(\id)$ and $\Q((1\dots r))$
we know that one of these $\Q(\alpha)$ is maximal over $\alpha\in Sym(r)$.\end{proof}
\begin{lem}
\label{lem:q-conjugacy-class}Let $\A$ be a two-rail channel, and
$[\beta]$ the conjugacy class of a permutation. Then $\Q$ restricted
to $[\beta]$ is maximal on non-overlapping members $\alpha\in[\beta]$.\end{lem}
\begin{proof}
Let $\alpha=\alpha_{1}\dots\alpha_{s}$ be a non-overlapping member
of the class and $\beta=\gamma\alpha\gamma^{-1}$ be any other member
of the class. First, remember

\begin{equation}
\Q(\alpha)=\sum_{\substack{\{x_{i}\}\\
i=1\dots r
}
}\Tr{\prod_{i=1}^{r}\A_{x_{i}x_{\alpha(i)}}}=\sum_{\delta_{1}\in\D(\alpha_{1})}\dots\sum_{\delta_{s}\in\D(\alpha_{s})}\Tr{\delta_{1}\dots\delta_{s}}\,.\label{eq:q-two-ways}
\end{equation}
With the first way of writing $\Q(\alpha)$ in mind we define a 1-1-mapping
between terms in the sum of $\Q(\alpha)$ and $\Q(\beta)$ via a mapping
of indices $x_{i}\to x_{\gamma^{-1}(i)}$ (or $x_{\gamma(i)}\to x_{i}$).
Then the products of matrices are mapped like
\begin{align*}
\A_{x_{1}x_{\alpha(1)}}\dots\A_{x_{r}x_{\alpha(r)}} & \to\A_{x_{\gamma(1)}x_{\gamma(\beta(1))}}\dots\A_{x_{\gamma(r)}x_{\gamma(\beta(r))}}\\
 & =\A_{x_{\gamma(1)}x_{\alpha(\gamma(1))}}\dots\A_{x_{\gamma(r)}x_{\alpha(\gamma(r))}}\,.
\end{align*}
In terms of diagrams this corresponds to permuting the {}``tiles''
(crossings or straight pieces) $\A_{x_{i}x_{\alpha(i)}}$ according
to the permutation $\gamma$. Some examples are shown in Appendix~\ref{app:q-conjugacy-class-example}.
In particular, this mapping does not change the number of any kind
of tile $\A_{00}$, $\A_{01}$, $\A_{10}$ or $\A_{11}$.

Now, consider the second way of writing $\Q(\alpha)$ in \eqref{eq:q-two-ways}
and split the sums over subdiagrams $\D(\alpha_{i})$ according to
$\D(\alpha_{i})=\D_{0}(\alpha_{i})\cup\D_{0}^{-1}(\alpha_{i})$
\[
\Q(\alpha)=\sum_{\delta_{1}\in\D_{0}(\alpha_{1})}\dots\sum_{\delta_{s}\in\D_{0}(\alpha_{s})}\sum_{\substack{\{t_{i}=\pm1\}\\
i=1\dots s
}
}\Tr{\delta_{1}^{t_{1}}\dots\delta_{s}^{t_{s}}}\,.
\]
We consider a subsum
\[
\sum_{\substack{\{t_{i}=\pm1\}\\
i=1\dots s
}
}\Tr{\delta_{1}^{t_{1}}\dots\delta_{s}^{t_{s}}}
\]
for fixed subdiagrams $\delta_{i}\in\D_{0}(\alpha_{i})$. In the following
we will prove that the contribution of this subsum to $\Q(\alpha)$
is larger or equal to the contribution of the subsum of the corresponding
diagrams to $\Q(\beta)$. From this it immediately follows that $\Q(\alpha)\geq\Q(\beta)$.

It follows a general proof of the inequality between corresponding
subsums. For illustration one subsum is evaluated in full detail with
diagrams in Appendix~\ref{app:q-conjugacy-class-example}.

For the non-overlapping permutation $\alpha$ in the subdiagrams $\delta_{i}$
all the straight lines and dashed lines are aligned, i.e. we have
a weak and a strong rail. Let $m_{i}$ be the number of crossings
and $n_{i}$ the number of straight pieces in $\delta_{i}$. The strong
rail in subdiagram $\delta_{i}$ contributes a factor $\lambda^{m_{i}}\mu^{n_{i}}$
and the weak rail contributes the factor $\kappa^{m_{i}}\nu^{n_{i}}$.
Summing over reflections the subsum equals
\[
2\prod_{i=1}^{s}\lambda^{m_{i}}\mu^{n_{i}}+\kappa^{m_{i}}\nu^{n_{i}}\,.
\]
On the other hand, for the possibly overlapping $\beta$ some of the
tiles are permuted and for one particular subdiagram, not all the
strong rail pieces might be on the same rail. Let $\tilde{m}_{i}$
be the number of crossings that are thus misaligned and $\tilde{n}_{i}$
the number of straight pieces that are misaligned. The two rails in
subdiagram $\delta_{i}$ now contribute the factors $\lambda^{m_{i}-\tilde{m}_{i}}\kappa^{\tilde{m}_{i}}\mu^{n_{i}-\tilde{n}_{i}}\nu^{\tilde{n}_{i}}$
and $\lambda^{\tilde{m}_{i}}\kappa^{m_{i}-\tilde{m}_{i}}\mu^{\tilde{n}_{i}}\nu^{n_{i}-\tilde{n}_{i}}$.
Summing over reflections the subsum adding to $\Q(\beta)$ equals
\[
2\prod_{i=1}^{s}\lambda^{m_{i}-\tilde{m}_{i}}\kappa^{\tilde{m}_{i}}\mu^{n_{i}-\tilde{n}_{i}}\nu^{\tilde{n}_{i}}+\lambda^{\tilde{m}_{i}}\kappa^{m_{i}-\tilde{m}_{i}}\mu^{\tilde{n}_{i}}\nu^{n_{i}-\tilde{n}_{i}}\,.
\]
Because $\lambda\geq\kappa$ and $\mu\geq\nu$ it follows that $\lambda^{m_{i}}\mu^{n_{i}}+\kappa^{m_{i}}\nu^{n_{i}}\geq\lambda^{m_{i}-\tilde{m}_{i}}\kappa^{\tilde{m}_{i}}\mu^{n_{i}-\tilde{n}_{i}}\nu^{\tilde{n}_{i}}+\lambda^{\tilde{m}_{i}}\kappa^{m_{i}-\tilde{m}_{i}}\mu^{\tilde{n}_{i}}\nu^{n_{i}-\tilde{n}_{i}}$
(the strong and weak rail dominate the two mixed rails), and hence
we have the desired inequality between corresponding contributions
to $\Q(\alpha)$ and $\Q(\beta)$.\end{proof}
\begin{prop}
For $\A=\Delta_{\lambda}$ the depolarizing channel in two dimensions,
$\Q(\alpha)$ is either maximal when $\alpha=(1\dots r)$ or when
$\alpha=\id$. \label{pro:maximal-q}\end{prop}
\begin{proof}
First, consider permutations that consist of non-overlapping $\alpha$,
Lemma~\ref{lem:non-overlap-max} proves that either $\alpha=\id$
or $\alpha=(1\dots r)$ yields the maximum $\Q(\alpha)$ amongst these
permutations. Every conjugacy class has a non-overlapping representant,
and Lemma~\ref{lem:q-conjugacy-class} states that these have maximal
$\Q$-value. Therefore either $\alpha=\id$ or $\alpha=(1\dots r)$
yield the maximal $\Q$-value amongst all the permutations $\alpha\in Symm(r)$.
\end{proof}

\subsection{Bound on Lipschitz constant\label{sub:lipschitz-constant}}

We derive an upper bound on the Lipschitz constant of the function
$f:S^{2^{n+1}-1}\to\R$ 
\[
f(\ket{\phi})=\Tr{\Deltaln(\ket{\phi}\bra{\phi})^{r}}
\]
 with respect to the Euclidean norm on $S^{2^{n+1}-1}\subset\R^{2^{n+1}}$.

We divide the function into four steps and prove bounds on the Lipschitz
constants for each step. The splitting is $f=d\circ c\circ b\circ a$
with
\begin{align*}
a & :\ket{\phi}\to\ket{\phi}\bra{\phi}=\rho\\
b & :\rho\rightarrow\Deltaln(\rho)=\rho'\\
c & :\rho'\rightarrow\textrm{eigenvalues of }\ensuremath{\rho'}=\vec{v}\\
d & :\vec{v}\rightarrow\sum_{i}v_{i}^{r}\,.
\end{align*}
Let $\M_{m}$ be the space of complex $m\times m$ matrices containing
the set of states $\S_{m}=\{\rho\in\M_{m}:\rho=\rho^{*},\tr\rho=1\}$.
\begin{lem}
Let $a(\ket{\phi})=\ket{\phi}\bra{\phi}$ a map $\C^{m}\to\M_{m}$
where $\braket{\phi}{\phi}=1$. Then the Lipschitz constant of $a$
with respect to the Euclidean norm in the domain and the Frobenius
norm $\|M\|_{2}=\Tr{MM^{*}}^{1/2}$ in the range is upper bounded
by $\sqrt{2}$.\label{lem:lipschitz-a}\end{lem}
\begin{proof}
For $\braket{\phi}{\phi}=\braket{\psi}{\psi}=1$ set $c=\braket{\phi}{\psi}$.
Now
\[
\|\ket{\phi}-\ket{\psi}\|_{2}^{2}=\braket{\phi}{\phi}-\braket{\phi}{\psi}-\braket{\psi}{\phi}+\braket{\psi}{\psi}=2-2\Re(c)\,,
\]
and
\[
\|\ket{\phi}\bra{\phi}-\ket{\psi}\bra{\psi}\|_{2}^{2}=\braket{\phi}{\phi}^{2}-2|\braket{\phi}{\psi}|^{2}+\braket{\psi}{\psi}^{2}=2-2|c|^{2}\,.
\]
Because of $\Re(c)\leq|c|$ and the inequality derived as follows
\begin{align*}
(1-|c|)^{2} & \geq0\\
2-2|c| & \geq1-|c|^{2}\\
2(2-2|c|) & \geq2-2|c|^{2}\,,
\end{align*}
we arrive at $2(2-2\Re(c))\geq2-2|c|^{2}$ or $\sqrt{2}\|\ket{\phi}-\ket{\psi}\|_{2}\geq\|\ket{\phi}\bra{\phi}-\ket{\psi}\bra{\psi}\|_{2}$.\end{proof}
\begin{lem}
Let $b(\rho)=\Deltaln(\rho)$ a map $\S_{2^{n}}\to\S_{2^{n}}$, then
the Lipschitz constant of $b$ with respect to the Frobenius norm
in domain and range is upper bounded by $\kappa^{n}$, where $\kappa=\lambda+\frac{1-\lambda}{\sqrt{2}}$
so $0<\kappa<1$ for $0<\lambda<1$.\label{lem:lipschitz-b}\end{lem}
\begin{proof}
It is useful to use the notation 
\[
\Deltaln=\sum_{J\subset\Z_{n}}\lambda^{n-|J|}(1-\lambda)^{|J|}\tr_{J}\otimes\left(\frac{\one}{2}\right)^{\otimes|J|}
\]
where $\tr_{J}$ is the partial trace over the systems with indices
in $J$. Notice that we use a loose notation of the tensor product
as the systems that are partially traced out and replaced by the totally
mixed states are not necessarily all on the right side of the tensor
product.

We will bound the operator norm $\|\A\|_{op}:=\sup_{\rho\in S}\frac{\|\A(\rho)\|_{2}}{\|\rho\|_{2}}$
where the supremum is over the set $S=\left\{ \ket{\phi}\bra{\phi}-\ket{\psi}\bra{\psi}\,\middle|\,\ket{\phi},\ket{\psi}\in\S_{2^{n}-1}\right\} $.
Because of 
\[
\|\A(\rho)-\A(\tau)\|_{2}=\|\A(\rho-\tau)\|_{2}\leq\|\A\|_{op}\|\rho-\tau\|_{2}
\]
for a linear map $\A$ bounding $\|\A\|_{op}$ immediately gives a
bound of the Lipschitz constant as well. Now 
\begin{align*}
\|\Deltaln\|_{op} & \leq\sum_{J}\lambda^{n-|J|}(1-\lambda)^{|J|}\left\Vert \tr_{J}\otimes\left(\frac{\one}{2}\right)^{\otimes|J|}\right\Vert _{op}\\
 & \leq\sum_{J}\lambda^{n-|J|}\fracp{1-\lambda}{\sqrt{2}}^{|J|}\\
 & =\sum_{k=0}^{n}{n \choose k}\lambda^{n-k}\fracp{1-\lambda}{\sqrt{2}}^{k}\\
 & =\left(\lambda+\frac{1-\lambda}{\sqrt{2}}\right)^{n}=\kappa^{n}\,,
\end{align*}
 where we used the bound 
\begin{align}
\left\Vert \tr_{\{1\dots k\}}\otimes\left(\frac{\one}{2}\right)^{\otimes k}\right\Vert _{op}^{2} & =\sup_{\rho\in S}\frac{\tr\left(\tr_{\{1\dots k\}}\rho\otimes\left(\frac{\one}{2}\right)^{\otimes k}\right)^{2}}{\tr\rho^{2}}\nonumber \\
 & =\tr\left(\left(\frac{\one}{2}\right)^{\otimes k}\right)^{2}\sup_{\rho\in S}\frac{\tr\left(\tr_{\{1\dots k\}}\rho\right)^{2}}{\tr\rho^{2}}\label{eq:last-supremum}\\
 & =\left(\frac{1}{2}\right)^{k}\,.\nonumber 
\end{align}

The supremum in \eqref{eq:last-supremum} was evaluated as follows.
First, consider that $\rho=\ket{\phi}\bra{\phi}-\ket{\psi}\bra{\psi}$
can be written as $\rho=\alpha\ket 0\bra 0-\alpha\ket 1\bra 1$ with
$0\leq\alpha\le1$, where $\ket 0$ and $\ket 1$ are orthonormal
states. Then the supremum runs over all possible orientations of $\ket 0$
and $\ket 1$
\begin{align*}
\sup_{\rho\in S}\frac{\tr\left(\tr_{\{1\dots k\}}\rho\right)^{2}}{\tr\rho^{2}} & =\sup_{\ket 0,\ket 1}\frac{\tr\left(\tr_{\{1\dots k\}}(\alpha\ket 0\bra 0-\alpha\ket 1\bra 1)\right)^{2}}{\tr(\alpha\ket 0\bra 0-\alpha\ket 1\bra 1)^{2}}\\
 & =\sup_{\ket 0,\ket 1}\frac{\tr\left(\tr_{\{1\dots k\}}(\ket 0\bra 0-\ket 1\bra 1)\right)^{2}}{\tr(\ket 0\bra 0-\ket 1\bra 1)^{2}}\\
 & =\sup_{\ket 0,\ket 1}\frac{\tr\left(\tr_{\{1\dots k\}}(\ket 0\bra 0-\ket 1\bra 1)\right)^{2}}{2}\,.
\end{align*}
Now assume $\rho_{0}=\tr_{\{1\dots k\}}\ket 0\bra 0$ and $\rho_{1}=\tr_{\{1\dots k\}}\ket 1\bra 1$
are arbitrary density matrices. Then
\begin{align*}
\sup_{\ket 0,\ket 1}\frac{\tr\left(\tr_{\{1\dots k\}}(\ket 0\bra 0-\ket 1\bra 1)\right)^{2}}{2} & =\sup_{\rho_{0},\rho_{1}}\frac{\Tr{\rho_{0}-\rho_{1}}^{2}}{2}\\
 & =\sup_{\rho_{0},\rho_{1}}\frac{\Tr{\rho_{0}^{2}+\rho_{1}^{2}-2\rho_{0}\rho_{1}}}{2}=1\,.
\end{align*}
The last equality follows from the fact, that $\tr\rho_{0,1}^{2}\leq1$
and $\tr\rho_{0}\rho_{1}\geq0$. The suprema are achieved when $\ket{\phi}=\ket{00}$
and $\ket{\psi}=\ket{11}$ so that $\rho_{0}=\ket 0\bra 0$ and $\rho_{1}=\ket 1\bra 1$.\end{proof}
\begin{rem}
Let $c:\S_{m}\rightarrow\R^{m}$ be the map that sends density matrices
to their eigenvalues, ordered high to low. The fact that the Lipschitz
constant of $c$ is upper bounded by $1$ is equivalent to the Hoffman-Wielandt
inequality \cite{hoffman_wielandt}
\[
\|\rho-\tau\|_{2}\geq\|c(\rho)-c(\tau)\|_{2}\,.
\]
\label{lem:lipschitz-c}\end{rem}
\begin{lem}
Let $d(\vec{v})=\sum_{i}v_{i}^{r}$ a map from $\{\vec{w}\in\R_{+}^{m}|\sum_{i}w_{i}=1\}$
to $\R$. Then the Lipschitz constant of $d$ is upper bounded by
$r$. \label{lem:lipschitz-d}\end{lem}
\begin{proof}
We have $\frac{\partial d}{\partial v_{j}}=rv_{j}^{r-1}$ so 
\[
\sup_{\vec{v}\in\mathrm{Dom}d}|\vec{\nabla}d|^{2}=\sup\sum r^{2}v_{i}^{2(r-1)}\leq r^{2}\,.
\]
 By integration we get 
\[
|d(\vec{v})-d(\vec{w})|\leq r\|\vec{v}-\vec{w}\|_{2}\,.
\]
\end{proof}
\begin{prop}
The Lipschitz constant $\eta$ of $\Tr{\Deltaln(\ket{\phi}\bra{\phi})^{r}}$,
with respect to the Euclidean norm in the domain, is upper bounded
by $\sqrt{2}r\kappa^{n}$, with $\kappa$ as in Lemma~\ref{lem:lipschitz-b}.\label{pro:lipschitz-constant}\end{prop}
\begin{proof}
From $\Tr{\Deltaln(\ket{\phi}\bra{\phi})^{r}}=d\circ c\circ b\circ a(\ket{\phi})$
with $a$, $b$, $c$ and $d$ as defined in Lemmas/Remark~\ref{lem:lipschitz-a}-\ref{lem:lipschitz-d}.
Thus, we simply combine the upper bounds of the lemmas and get the
bound $\sqrt{2}\cdot\kappa^{n}\cdot1\cdot r$.
\end{proof}

\section{Other results\label{sec:Other-results}}

\subsection{$\Q((1\dots r))$ for a more general qubit channel}

For a channel $\A$ that maps that scales the Bloch sphere like 
\[
\vec{n}\rightarrow\begin{pmatrix}\lambda_{1} & 0 & 0\\
0 & \lambda_{2} & 0\\
0 & 0 & \lambda_{3}
\end{pmatrix}\vec{n}\,,
\]
the Choi-Jamiolkowski representation is
\[
Choi(\Delta_{\lambda})=\begin{pmatrix}\mu & 0 & 0 & \lambda\\
0 & \nu & \kappa & 0\\
0 & \kappa & \nu & 0\\
\lambda & 0 & 0 & \mu
\end{pmatrix}\,,
\]
i.e. this is a two-rail channel with $\mu=\frac{1+\lambda_{3}}{2}$,
$\nu=\frac{1-\lambda_{3}}{2}$, $\lambda=\frac{\lambda_{1}+\lambda_{2}}{2}$
and $\kappa=\frac{\lambda_{1}-\lambda_{2}}{2}$. Therefore we have
\[
\Q(\id)=\tr\A(\one)^{r}=\tr\one=2
\]
and according to \eqref{eq:q-1tor-dualrail}
\begin{align*}
\Q((1\dots r)) & =\left(\mu+\lambda\right)^{r}+\left(\mu-\lambda\right)^{r}+\left(\nu+\kappa\right)^{r}+\left(\nu-\kappa\right)^{r}\\
 & =\left(\frac{1+\lambda_{1}+\lambda_{2}+\lambda_{3}}{2}\right)^{r}+\left(\frac{1-\lambda_{1}-\lambda_{2}+\lambda_{3}}{2}\right)^{r}\\
 & +\fracp{1+\lambda_{1}-\lambda_{2}-\lambda_{3}}2^{r}+\fracp{1-\lambda_{1}+\lambda_{2}-\lambda_{3}}2^{r}\,.
\end{align*}
Assuming $\Qmax=\max\{2,\Q((1\dots r))\}$ the output is maximally
mixed if $|\lambda_{1}|+|\lambda_{2}|+|\lambda_{3}|\leq1$. But this
is exactly the condition for $\A$ to be entanglement breaking, according
to Theorem 3 in \cite{ruskai_entanglementbreaking}.

\subsection{$\Q((1\dots r))$ of $\Delta_{\lambda}$ for any dimension $d$\label{sec:q-dim-d}}

For the $d$-dimensional depolarizing channel we have $\Q(\id)=\tr\one=d$
as usual. Furthermore, the Choi-Jamiolkowsi representation consists
of two blocks, one $d$-dimensional block with diagonal entries $\mu=\frac{1+\lambda}{d}$
and off-diagonal entries $\lambda$, and another block that is $\nu=\frac{1-\lambda}{d}$
times identity on the other $d^{2}-d$ dimensions. For example for
$d=3$ the representation looks like
\[
Choi(\Delta_{\lambda})=\begin{pmatrix}\mu &  &  &  & \lambda &  &  &  & \lambda\\
 & \nu\\
 &  & \nu\\
 &  &  & \nu\\
\lambda &  &  &  & \mu &  &  &  & \lambda\\
 &  &  &  &  & \nu\\
 &  &  &  &  &  & \nu\\
 &  &  &  &  &  &  & \nu\\
\lambda &  &  &  & \lambda &  &  &  & \mu
\end{pmatrix}\,.
\]
The eigenvalues of the first block are $\mu-\lambda$ with $(d-1)$-multiplicity
and a single eigenvalue $\mu+(d-1)\lambda$. Therefore,
\begin{align*}
\Q((1\dots r)) & =\tr Choi(\Delta_{\lambda})^{r}\\
 & =(\mu+(d-1)\lambda)^{r}+(d-1)(\mu-\lambda)^{r}+(d^{2}-d)\nu^{r}\\
 & =\fracp{1+(d^{2}-1)\lambda}d^{r}+(d^{2}-1)\fracp{1-\lambda}d^{r}\,.
\end{align*}
This result agrees with the result for $d=2$ found in \ref{sub:evaluating-q}.
The critical value is $\lambda=\frac{1}{d+1}$, below this value the
average output is maximally mixed.

\section{Conclusion}

We found a general explicit form for $\beta_{r}(\A)$ depending on
the function $\Q_{\A}:Sym(r)\to\R$. In the limit $n\to\infty$ the
maximal term $\Q_{\max}$ is dominant in $\beta_{r}^{\reg}(\A)$ and
therefore the only relevant term. However, finding $\Qmax$ is not
easy in general. In the case of the qubit depolarizing channel we
proved that $\Q_{\max}=\max\{2,\Q_{\Delta_{\lambda}}((1\dots r))\}$
and also that $\overline{\beta}_{r}(\Delta_{\lambda})=\Srreg(\Delta_{\lambda})$
for some $\lambda$. For all $r\in\N$ and $\lambda\leq1/3$ the regularized
output entropy becomes 1. Because the typical high-dimensional random
state is highly entangled our result for $\lambda\leq1/3$ agrees
with the notion that the tensor product of an entanglement breaking
channel {}``chops up'' these highly entangled states and produces
maximally mixed states with almost certain probability.

For future work it would be interesting to also study channels other
than the depolarizing channel, especially channels that are not known
to be additive. Also, it might be of interest to study quantities
similar to $\Srreg$ with the same general procedure, for example
$\E{\tr\A\otimes\one(\ket{\phi}\bra{\phi})^{r}}$. It would be insightful
to gain a better understanding of the typical output and input states
by finding explicit examples that conform with the average output.

\section{Acknowledgements}

We thank the reviewer for important corrections and improvements which
helped the quality and presentation of the paper.

\appendix

\section{$\Q$ sum diagrams\label{app:q-conjugacy-class-example}}

We calculate an example of the subsums appearing in the proof of Lemma~\ref{lem:q-conjugacy-class}.
Let $\alpha=(123)(45)$, $\beta=(143)(25)$ and $\gamma=(24)$. With
this choice the correspondence of diagrams consists of switching tiles
2 and 4. The choice of indices 
\[
\begin{array}{|ccccc|}
\hline x_{1} & x_{2} & x_{3} & x_{4} & x_{5}\\
\hline 0 & 1 & 1 & 0 & 0
\\\hline \end{array}
\]
 gives the diagrams and totals shown in Table~\ref{tab:contributions}.
The contribution to $\Q(\alpha)$ dominates in both factors as expected
because $\alpha$ is non-overlapping.

\begin{table}
\begin{centering}
\includegraphics{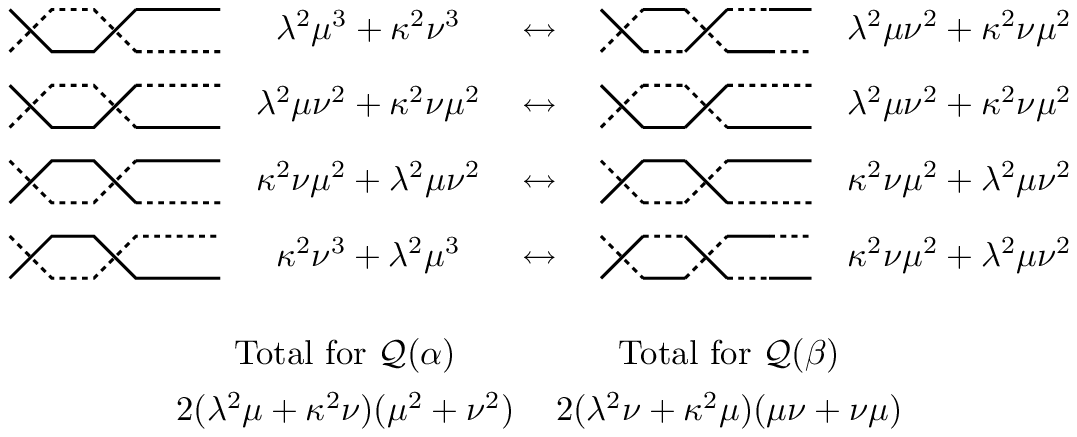}
\par\end{centering}

\caption{Contributions to the sums $\Q(\alpha)$ and $\Q(\beta)$.\label{tab:contributions}}
\end{table}
\bibliographystyle{ieeetr}

\end{document}